\documentclass[12pt]{iopart}

\usepackage{epsfig,amssymb,iopams,bbm,amsthm,verbatim}

\newcommand{\aff}[1]{\boldsymbol{\mathcal{#1}}}
\newcommand{\vecttt}[3]{\left( \! \begin{array}{c} #1 \\ #2 \\ #3  \end{array} \! \right)}
\newcommand{\vectt}[2]{\left( \! \begin{array}{c} #1 \\ #2  \end{array} \! \right)}

\newcommand{\id}{\mathbbm{1}}

\newcommand{\asym}[1]{\mathcal{S}\left( #1 \right)}
\newcommand{\asyms}[1]{\mathcal{S}( #1)}
\newcommand{\abs}[1]{\left\vert #1 \right\vert}

\newtheorem{lemma}{Lemma}
\newtheorem{corollary}{Corollary}

\newtheorem{proposition}{Proposition}
\newtheorem{theorem}{Theorem}

\begin{document}

\today

\title[Multi-terminal Thermoelectric Transport in a Magnetic
Field]{Multi-terminal Thermoelectric Transport in a Magnetic
Field:\\ Bounds on Onsager Coefficients and Efficiency}

\author{Kay Brandner and Udo Seifert}

\address{II. Institut f\"ur Theoretische Physik, Universit\"at Stuttgart, 70550 Stuttgart, Germany}

\begin{abstract}
Thermoelectric transport involving an arbitrary number of terminals is 
discussed in the presence of a magnetic field breaking time-reversal 
symmetry within the linear response regime using the Landauer-B\"uttiker
formalism. We derive a universal bound on the Onsager coefficients that 
depends only on the number of terminals. This bound implies bounds on 
the efficiency and on efficiency at maximum power for heat engines and 
refrigerators. For isothermal engines pumping particles and for absorption
refrigerators these bounds become independent even of the number of 
terminals. On a technical level, these results follow from an original 
algebraic analysis of the asymmetry index of doubly substochastic matrices
and their Schur complements. 
\end{abstract}

\pacs{72.15.Jf, 05.70.Ln}
\submitto{\NJP}
\maketitle

\section{Introduction}

Thermoelectric devices use a coupling between heat and particle currents 
driven by local gradients in temperature and chemical potential to
generate electrical power or for cooling \cite{Dresselhaus2007,Snyder2008,Bell2008,Vineis2010}. Since they 
work without any moving parts, such machines have a lot 
of advantages compared to their cyclic counterparts, which rely on the 
periodic compression and expansion of a certain working fluid 
\cite{Humphrey2005a}. However, so far their notoriously modest efficiency 
prevents a wide-ranging applicability. Although it has been shown that
proper energy filtering leads to highly efficient 
thermoelectric heat engines \cite{Mahan1996},
which, in principle, may even reach Carnot efficiency
\cite{Humphrey2002, Humphrey2005}, so far no competitive devices 
coming even close to this limit are available. Consequently,
the challenge of finding better thermoelectric materials has attracted
a great amount of scientific interest during the last decades.\\

Recently, Benenti \etal discovered a new option to enhance the 
performance of thermoelectric engines \cite{Benenti2011}.
Their rather general analysis
within the phenomenological framework of linear irreversible
thermodynamics reveals that a magnetic field, which breaks time 
reversal symmetry, could enhance thermoelectric efficiency
significantly. In principle, it even seems to be possible to get
completely reversible transport, i.e., devices that work at Carnot
efficiency while delivering finite power output. This spectacular 
observation prompts the question, whether this option can be realized
in specific microscopic models.\\

An elementary and well established framework for the description of 
thermoelectric transport on a microscopic level is provided by the 
scattering approach  originally pioneered by Landauer \cite{Landuer1957}. 
The basic idea behind this method is to connect two electronic reservoirs (terminals) of different temperature and chemical potential via perfect, infinitely long leads to a central scattering region. By assuming non
interacting electrons, which are transferred coherently between the
terminals, it is possible to express the linear transport coefficients
in terms of the scattering matrix that describes the motion of a single 
electron of energy $E$ through the central region. Thus, the 
macroscopic transport process can be traced back to the microscopic 
dynamics of the electrons. This formalism can easily be extended to
an arbitrary number of terminals \cite{Sivan1986, Butcher1990}.\\

Within a purely coherent two-terminal set-up,
current conservation requires a symmetric scattering matrix and hence
a symmetric matrix of kinetic coefficients,
even in the presence of a magnetic field \cite{Buttiker1988a}. Therefore,
without inelastic scattering events the broken time reversal symmetry
is not visible on the macroscopic scale. An elegant way to simulate
inelastic scattering within an inherently conservative system goes back to B\"uttiker \cite{Buttiker1988}. He proposed to attach additional, 
so-called probe terminals to the scattering region, whose temperature
and chemical potential are adjusted in such a way that they do not exchange 
any net quantities with the remaining terminals but only induce
phase-breaking.\\

The arguably most simple case is to include only one probe terminal, which
leads to a three-terminal model. Saito \etal \cite{Saito2011} pointed out 
that such a minimal set-up is sufficient to obtain a non-symmetric matrix of kinetic coefficients. However, we have shown in a preceding work on the
three-terminal system \cite{Brandner2013} that current conservation puts 
a much stronger bound on the Onsager coefficients than the bare second law.
It turned out that this new bound constrains the maximum efficiency of 
the model as a thermoelectric heat engine to be significantly smaller than 
the Carnot value as soon as the Onsager matrix becomes non-symmetric. 
Moreover, Balachandran \etal \cite{Balachandran2013} demonstrated by 
extensive numerical efforts that our bound is tight.\\

The strong bounds on Onsager coefficients and efficiency obtained within
the three-terminal set-up raise the question whether they persist if more
terminals are included. This problem will be addressed in this paper. 
We will derive a universal bound on kinetic coefficients that depends only 
on the number of terminals and gets weaker as this number increases. 
Only in the limit of infinitely many terminals, this bound approaches 
the well-known one following from the positivity of entropy production.
By specializing these results to thermoelectric transport between two 
\emph{real} terminals with the other $n-2$ acting as probe terminals, we obtain
bounds on the efficiency and the efficiency at maximum power for different
variants of thermoelectric devices like heat engines and cooling devices.\\

Our results follow from analyzing the matrix of kinetic coefficients in the 
$n$-terminal set-up and its subsequent specializations to two real and 
$n-2$ probe terminals. On a technical level, we introduce an asymmetry index
for a positive semi-definite matrix and compute it for the class of matrices
characteristic for the scattering approach. These calculations involve a
fair amount of original matrix algebra for doubly substochastic matrices 
and their Schur complements, which we develop in an extended and self-contained
mathematical appendix.\\

The main part of the paper is organized as follows. In section 2, we 
introduce the multi-terminal model and recall the expressions for its kinetic
coefficients. In section 3, we derive the new bounds on these coefficients.
In section 4, we show how these bounds imply bounds on the efficiency and the
efficiency at maximum power for heat engines, for refrigerators, for iso-thermal
engines and for absorption refrigerators. In contrast to the former two 
classes, the latter two involve only one type of affinities, namely chemical 
potential or temperature differences, respectively, which implies even stronger
bounds. We conclude in section 5.

\newpage

\section{The Multi-terminal Model}

\begin{figure}[htb]
\centering
\epsfig{file=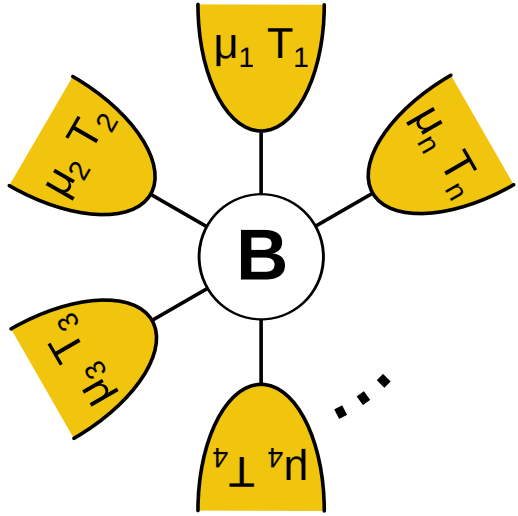, scale=1.7}
\caption{Sketch of the multi-terminal model for 
thermoelectric transport
\label{Fig_Multi_Terminal_Model}}
\end{figure}

We consider the set-up schematically shown in figure 
\ref{Fig_Multi_Terminal_Model}. A central scattering region
equipped with a constant magnetic field $\mathbf{B}$ is 
connected to $n$ independent electronic reservoirs
(terminals) of respective temperature $T_1, \dots, T_n$ 
and chemical potential $\mu_1, \dots \mu_n$. We assume 
non interacting electrons, which are transferred coherently
between the terminals without any inelastic scattering.
In order to describe the resulting transport
process within the framework of linear irreversible 
thermodynamics, we fix the reference temperature 
$T \equiv T_1$ and chemical potential $\mu \equiv \mu_1$,
and define the affinities 
\begin{eqnarray}\label{MT Affinities}
\mathcal{F}^{\rho}_{\alpha} \equiv \frac{\mu_{\alpha}-\mu}{T} 
\equiv \frac{\Delta\mu_{\alpha}}{T} \qquad {{ \rm and }} \qquad
\mathcal{F}^q_{\alpha} \equiv \frac{T_{\alpha} -T}{T^2} \equiv 
\frac{\Delta T_{\alpha}}{T^2}\\ \left( \alpha = 2, \dots n
\right).
\nonumber
\end{eqnarray}
By $J^{\rho}_{\alpha}$ and $J^q_{\alpha}$ we denote the charge
and the heat current flowing out of the reservoir $\alpha$,
respectively.
Within the linear response regime, which is valid as long as
the temperature and chemical potential differences 
$\Delta T_{\alpha}$ and $\Delta \mu_{\alpha}$ are small compared
to the respective reference values, the currents and affinities 
are connected via the phenomenological equations
\cite{Callen1985}
\begin{equation}\label{Phenomenological Equations}
\mathbf{J} = \mathbb{L}(\mathbf{B}) \aff{F}.
\end{equation}
Here, we introduced the current vector 
\begin{equation}
\mathbf{J}
\equiv \vecttt{\mathbf{J}_2}{\vdots}{\mathbf{J}_n}
\qquad {{ \rm and}} \; {{\rm the }} \; {{\rm affinity }} \;
{{ \rm vector}} \qquad
\aff{F}=\vecttt{\aff{F}_2}{\vdots}{\aff{F}_n}
\end{equation}
with the respective subunits
\begin{equation}
\mathbf{J}_{\alpha} \equiv \vectt{J^{\rho}_{\alpha}}
{J^q_{\alpha}} 
\qquad {{\rm and}} \quad
\aff{F}_{\alpha} \equiv 
\vectt{\mathcal{F}_{\alpha}^{\rho}}{\mathcal{F}_{\alpha}^q}.
\end{equation}
Analogously, we divide the matrix of kinetic coefficients 
\begin{equation}\label{Onsager Matrix}
\mathbb{L}(\mathbf{B}) \equiv 
\left( \! \begin{array}{ccc}
\mathbb{L}_{22}(\mathbf{B}) & \cdots & 
\mathbb{L}_{2n}(\mathbf{B})\\
\vdots           & \ddots & \vdots\\
\mathbb{L}_{n2}(\mathbf{B})  & \cdots & 
\mathbb{L}_{nn}(\mathbf{B})
\end{array} \! \right) \in \mathbb{R}^{2(n-1)\times 2(n-1)}
\end{equation}
into the $2 \times 2$ blocks $\mathbb{L}_{\alpha \beta} \in 
\mathbb{R}^{2 \times 2}$ ($\alpha,\beta = 2,\dots, n$), which can
be calculated explicitly. By making use of the 
multi-terminal Landauer formula  \cite{Sivan1986, Butcher1990},
we get the expression
\begin{equation}\label{Landauer Buttiker}
\mathbb{L} _{\alpha\beta}(\mathbf{B})
 = \frac{Te^2}{h} \int_{- \infty}^{\infty} dE \;
F(E) \left( \! \begin{array}{cc}
1 & \frac{E-\mu}{e} \\
\frac{E-\mu}{e} & \left( \frac{E-\mu}{e} \right)^2
\end{array} \! \right)
\left( \delta_{\alpha\beta} - T_{\alpha\beta}(E, \mathbf{B} )
\right),
\end{equation}
where $h$ denotes Planck's constant, $e$ the electronic
unit charge,
\begin{equation}
F(E) \equiv \left[ 4 k_B T \cosh^2 \left(\frac{E-\mu}{2 k_BT}
\right) \right]^{-1}
\end{equation}
the negative derivative of the Fermi function and $k_B$ 
Boltzmann's constant.\\

The expression (\ref{Landauer Buttiker}) shows that the
transport properties of the model are completely determined by
the transition probabilities $T_{\alpha\beta}(E,\mathbf{B})$, 
which obey two important relations. First, current conservation
requires the sum rule 
\begin{equation}\label{Sum Rule for T}
\sum_{\alpha=1}^n T_{\alpha\beta}(E, \mathbf{B}) = 
\sum_{\beta=1}^n T_{\alpha\beta}(E, \mathbf{B}) =1,
\end{equation}
i.e., the transition matrix 
\begin{equation}\label{Transition Matrix}
\mathbb{T}(E,\mathbf{B}) \equiv \left( \!
\begin{array}{ccc}
T_{11}(E, \mathbf{B}) & \cdots & T_{1n}(E, \mathbf{B})\\
\vdots                & \ddots & \vdots\\
T_{n1}(E, \mathbf{B}) & \cdots & T_{nn}(E, \mathbf{B})
\end{array} \! \right) \in \mathbb{R}^{n \times n}
\end{equation} 
is doubly stochastic for any $E \in\mathbb{R}$ and $\mathbf{B}\in
\mathbb{R}^3$.
Second, due to time reversal symmetry, the
$T_{\alpha\beta}(E, \mathbf{B})$ have to posses the symmetry
\begin{equation}\label{Time Reversal Symmetry for T}
T_{\alpha\beta}(E, \mathbf{B}) = T_{\beta\alpha}(E, -\mathbf{B}).
\end{equation}
Notably, for a fixed magnetic field $\mathbf{B}$, the transition 
matrix $\mathbb{T}(E,\mathbf{B})$ does not necessarily have to be
symmetric. This observation will be crucial for the subsequent 
considerations.\\

For later purpose, we note that,
by combining (\ref{Onsager Matrix}) and (\ref{Landauer Buttiker}),
$\mathbb{L}(\mathbf{B})$ can be expressed as an integral over
tensor products given by 
\begin{equation}\label{MT Landauer Buttiker Tensor Product}
\mathbb{L}(\mathbf{B}) = \frac{Te^2}{h}\int_{-\infty}^{\infty}
dE \; F(E) \left( \id -\bar{\mathbb{T}}(E,\mathbf{B})\right)
\otimes \left( \! \begin{array}{cc}
1 & \frac{E-\mu}{e} \\
\frac{E-\mu}{e} & \left( \frac{E-\mu}{e} \right)^2
\end{array} \! \right).
\end{equation}
Here, $\id$ denotes the identity matrix and  $\bar{\mathbb{T}}(E,
\mathbf{B})$ arises from $\mathbb{T}(E,\mathbf{B})$ by deleting 
the first row and column. Consequently, the matrix $\bar{\mathbb{
T}}(E,\mathbf{B})$ must be \emph{doubly substochastic}, which means
that all entries of $\bar{\mathbb{T}}(E,\mathbf{B})$ are 
non-negative and any row and column sums up to a value \emph{not
greater} than $1$.

\section{Bounds on the Kinetic Coefficients}

\subsection{Phenomenological Constraints}

The phenomenological framework of linear irreversible 
thermodynamics provides two fundamental constraints on the
matrix of kinetic coefficients $\mathbb{L}(\mathbf{B})$.
First, since the entropy  production accompanying the 
transport process descibed by
(\ref{Phenomenological Equations}) reads \cite{Callen1985}
\begin{equation}\label{Entropy Production}
\dot{S} = \aff{F}^t \mathbf{J} 
= \aff{F}^t\mathbb{L}(\mathbf{B}) \aff{F},
\end{equation}
the second law requires $\mathbb{L}(\mathbf{B})$ to be 
positive semi-definite.
Second, Onsager's reciprocal relations impose the symmetry 
\begin{equation}\label{Time Reversal Symmetry for L}
\mathbb{L}^t(\mathbf{B}) = \mathbb{L}(- \mathbf{B}). 
\end{equation}
Apart from these constraints, no further general relations
restricting the elements of $\mathbb{L}(\mathbf{B})$ at fixed
magnetic field $\mathbf{B}$ are known. We will now 
demonstrate that such a lack of constraints leads to profound
consequences for the thermodynamical properties of this model.
To this end, we split the current vector 
$\mathbf{J}$ into an irreversible and a reversible part 
given by 
\begin{equation}
\mathbf{J}^{{{\rm irr}}} \equiv 
\frac{\mathbb{L}(\mathbf{B}) + \mathbb{L}^t(\mathbf{B})}{2}
\aff{F}
\qquad {{\rm and }} \qquad 
\mathbf{J}^{{{\rm rev}}} \equiv 
\frac{\mathbb{L}(\mathbf{B}) - \mathbb{L}^t(\mathbf{B})}{2}
\aff{F}
\end{equation}
respectively. The reversible part vanishes for 
$\mathbf{B}=\boldsymbol{0}$ by virtue of the reciprocal 
relations (\ref{Time Reversal Symmetry for L}). However, 
in situations with $\mathbf{B} \neq \boldsymbol{0}$ it can 
become arbitrarily large without contributing to the 
entropy production (\ref{Entropy Production}). In 
principle, it would be even possible to have $\dot{S} =0$ 
and $\mathbf{J}^{{{\rm rev}}} \neq \boldsymbol{0}$ 
simultaneously, i.e., completely reversible transport, 
suggesting \emph{inter alia} the opportunity for a 
thermoelectric heat engine operating at Carnot efficiency
with finite power output \cite{Benenti2011}. This observation
raises the question, whether there might be 
stronger relations between the kinetic coefficients
going beyond the well known reciprocal relations 
(\ref{Time Reversal Symmetry for L}). In the next section, starting
from the microscopic representation (\ref{Landauer Buttiker}),
we derive bounds on the kinetic coefficients, which prevent 
this option of Carnot efficiency with finite power.

\subsection{Bounds following from Current Conservation}

These bounds can be derived by first quantifying the asymmetry
of the Onsager matrix  $\mathbb{L}(\mathbf{B})$.
For an arbitrary positive semi-definite matrix $\mathbb{A} \in
\mathbb{R}^{m \times m}$ we define an asymmetry index by 
\begin{equation}\label{Asymmetry Index}
\mathcal{S}(\mathbb{A}) \equiv 
\min \left\{ \left. s \in \mathbb{R} \right\vert 
\forall \mathbf{z} \in \mathbb{C}^{m} \; \; \;
\mathbf{z}^{\dagger}
\left(s \left( \mathbb{A} + \mathbb{A}^t \right)
+ i \left( \mathbb{A}
-  \mathbb{A}^t \right) \right)\mathbf{z} \geq 0
\right\}.
\end{equation}
Some of the basic properties of this asymmetry index are outlined 
in \ref{Apx AI Basic properties}. We note that a quite similar quantity
was introduced by Crouzeix and Gutan \cite{Crouzeix2003} in another
context.\\

We will now proceed in two steps. First, we show that the
asymmetry index of the matrix of kinetic coefficients 
$\mathbb{L}(\mathbf{B})$ and all its principal submatrices
is bounded from above for any finite number of terminals $n$.
Second, we will derive therefrom a set of new bounds 
on the elements of $\mathbb{L}(\mathbf{B})$, which go beyond the
second law. We note that from now on we notationally suppress 
the dependence of any quantity on the magnetic field in
order to keep the notation slim. \\

For the first step, we define the quadratic form
\begin{equation}\label{Asymmetry Form}
Q(\mathbf{z},s) \equiv \mathbf{z}^{\dagger} 
\left(s\left(\mathbb{L}_A + \mathbb{L}^t_A \right)
+ i\left( \mathbb{L}_A - \mathbb{L}^t_A \right)
\right)\mathbf{z}
\end{equation}
for any $\mathbf{z} \in \mathbb{C}^{2m}$ and any 
$s\in\mathbb{R}$. Here, $A\subset \left\{2,\dots, n\right\}$
denotes a set of $m\leq n-1$ integers. The matrix $\mathbb{L}_A$ 
arises from $\mathbb{L}$ by taking all blocks $\mathbb{L}_{
\alpha\beta}$ with column and row index in $A$, i.e., $\mathbb{L
}_A$ is a principal submatrix of $\mathbb{L}$, which 
preserves the $2 \times 2$ block structure shown in
(\ref{Onsager Matrix}). Comparing (\ref{Asymmetry Form}) with the
definition (\ref{Asymmetry Index}) reveals that the minimum $s$
for which $Q(\mathbf{z},s)$ is positive semi-definite equals 
the asymmetry index of $\mathbb{L}_A$. 
Next, by recalling (\ref{MT Landauer Buttiker Tensor Product}) we 
rewrite the matrix $\mathbb{L}_A$ in
the rather compact form
\begin{equation}\label{Tensor Representation OM}
\mathbb{L}_A
= \frac{T e^2}{h} \int_{-\infty}^{\infty} dE \; F(E)
\left(\mathbbm{1} - \bar{\mathbb{T}}_A
(E) \right)\otimes \left( \! \begin{array}{cc}
1 & \frac{E-\mu}{e} \\
\frac{E-\mu}{e} & \left( \frac{E-\mu}{e} \right)^2
\end{array} \! \right),
\end{equation}
where $\bar{\mathbb{T}}_A (E) \in \mathbb{R}^{m\times m}$ is
obtained from $\bar{\mathbb{T}}(E)$ by taking the rows and
columns indexed by the set $A$.
Decomposing the vector $\mathbf{z}$ as
\begin{equation}\label{Decomposition of z}
\mathbf{z} \equiv \mathbf{z}_1 \otimes \vectt{1}{0}
+ \mathbf{z}_2 \otimes \vectt{0}{1} 
\qquad {{\rm with}} \qquad
\mathbf{z}_1, \mathbf{z}_2 \in \mathbb{C}^m
\end{equation}
and inserting (\ref{Tensor Representation OM})
and (\ref{Decomposition of z}) into (\ref{Asymmetry Form}) 
yields 
\begin{equation}\label{Asymmetry Form K}
Q(\mathbf{z},s) = \frac{T e^2}{h}
\int_{-\infty}^{\infty} dE \; F(E)
\mathbf{y}^{\dagger}(E)\mathbb{K}(E,s)
\mathbf{y}(E).
\end{equation}
Here we introduced the vector 
\begin{equation}
\mathbf{y}(E) \equiv \mathbf{z}_1  + \frac{E-\mu}{e}
\mathbf{z}_2
\end{equation}
and the Hermitian matrix 
\begin{eqnarray}\label{Auxiliary Matrix}
\mathbb{K}_A(E,s) \equiv s \left( 2  \cdot \mathbbm{1} 
- \bar{\mathbb{T}}_A(E) 
- \bar{\mathbb{T}}^t_A(E) \right)
- i \left( \bar{\mathbb{T}}_A(E) 
- \bar{\mathbb{T}}^t_A(E) \right)
\in \mathbb{C}^{m\times m}, \nonumber\\
\end{eqnarray}
which is positive semi-definite for any 
\begin{equation}
s \geq \asym{ \id - \bar{\mathbb{T}}_A(E)}.
\end{equation}
However, since $\bar{\mathbb{T}}(E)$ is doubly stochastic for any $E$,
the matrix $\bar{\mathbb{T}}_A(E)$ must have the same property
and it follows from Corollary 
\ref{Apx AI Doubly Substochastic Matrices} proven in 
\ref{Apx Bound on AI for special classes of matrices}
\begin{equation}
\asym{ \id - \bar{\mathbb{T}}_A(E) } \leq
\cot \left( \frac{ \pi}{m+1} \right).
\end{equation}
Hence, independently of $E$, $\mathbb{K}_A(E,s)$ is positive 
semi-definite for any 
\begin{equation}\label{Bound on asymmetry index DS matrices}
s \geq \cot \left( \frac{\pi}{m+1} \right). 
\end{equation}
Finally, we can infer from (\ref{Asymmetry Form K}) that 
$Q(\mathbf{z},s)$ is positive semi-definite for any
$s$, which obeys (\ref{Bound on asymmetry index DS matrices}).
Consequently, with (\ref{Asymmetry Form}), we have the desired
bound on the asymmetry index of $\mathbb{L}_A$ as 
\begin{equation}\label{Asymmetry Bound}
\asym{\mathbb{L}_A}\leq \cot \left( \frac{\pi}{m+1} \right).
\end{equation} 
This bound, which ultimately follows from current conservation,
constitutes our first main result.\\

\begin{figure}[t]
\centering 
\epsfig{file=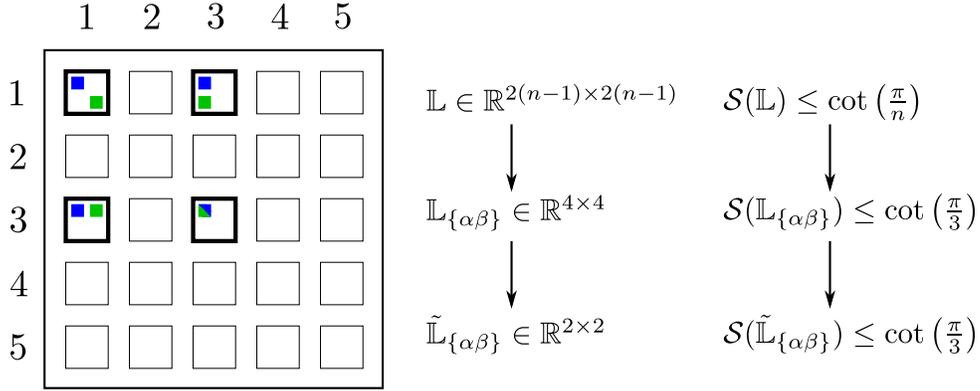, scale=2}
\caption{Schematic illustration of the reduction from $\mathbb{L}$ to 
$\tilde{\mathbb{L}}_{\alpha\beta}$. The big square represents 
$\mathbb{L}$ for the case $n=6$, the smaller ones correspond to the
$2\times 2$ blocks 
introduced in (\ref{Onsager Matrix}). By taking the bold framed 
squares, the $4\times 4$
matrix $\mathbb{L}_{\{\alpha\beta\}}$ is obtained for the case
$\alpha=1$ and $\beta=3$. The filled squares represent the elements of
the $2\times 2$ matrix $\tilde{\mathbb{L}}_{\{\alpha\beta\}}$ introduced in 
(\ref{B ij-Matrix}) for $(i,j)=(1,1)$ (blue) and 
$(i,j)=(2,1)$ (green). \label{Fig Matrix Reduction}}
\end{figure}

We will now demonstrate that (\ref{Asymmetry Bound}) puts
indeed strong bounds on the kinetic coefficients. To this end,
we extract a $2\times 2$ principal submatrix from $\mathbb{L}$ 
by a two-step procedure, which is schematically summarized in 
figure \ref{Fig Matrix Reduction}. In the first step, we consider the
$4\times 4$ principal submatrix of $\mathbb{L}$ given by 
\begin{equation}\label{LA2 Matrix}
{\mathbb{L}}_{\{\alpha,\beta\}} \equiv \left( \! 
\begin{array}{cc}
\mathbb{L}_{\alpha\alpha} & \mathbb{L}_{\alpha\beta}\\
\mathbb{L}_{\beta\alpha} & \mathbb{L}_{\beta\beta} 
\end{array} \! \right),
\end{equation}
which arises from $\mathbb{L}$ by taking only the blocks with row and
column index equal to $\alpha$ or $\beta$. From 
(\ref{Asymmetry Bound}) we immediately get with $m=2$
\begin{equation}\label{AI LA2 Matrix}
\asym{\mathbb{L}_{\{\alpha,\beta\}}}\leq\cot\left( \frac{\pi}{3}
\right)= \frac{1}{\sqrt{3}}. 
\end{equation}
Next, from (\ref{LA2 Matrix}), we take a $2\times 2$ principal
submatrix 
\begin{equation}\label{B ij-Matrix}
\tilde{\mathbb{L}}_{\{\alpha,\beta\}}\equiv
\left( \! \begin{array}{cc}
  \left(\mathbb{L}_{\alpha\alpha}\right)_{ii} 
& \left(\mathbb{L}_{\alpha\beta}\right)_{ij}\\
  \left(\mathbb{L}_{\beta\alpha}\right)_{ji}
& \left(\mathbb{L}_{\beta\beta}\right)_{jj}
\end{array} \! \right)
\equiv 
\left( \! \begin{array}{cc}
  L_{11} & L_{12}\\
  L_{21} & L_{22}
\end{array} \! \right),
\end{equation}
where $\left(\mathbb{L}_{\alpha\beta}\right)_{ij}$ with $i,j=1,2$ 
denotes the $(i,j)$-entry of the block matrix $\mathbb{L}_{
\alpha\beta}$. By virtue of Proposition
\ref{Apx Dominance of Principal Submatrices AI} proven in
\ref{Apx Bound on AI for special classes of matrices}, the inequality
(\ref{AI LA2 Matrix}) implies 
\begin{equation}
\asym{\tilde{\mathbb{L}}_{\{\alpha,\beta\}}}
\leq \frac{1}{\sqrt{3}},
\end{equation}
which is equivalent to requiring the Hermitian matrix
\begin{equation}
\tilde{\mathbb{K}}_{ \{\alpha,\beta \} }  \equiv 
\frac{1}{\sqrt{3}} \left(
\tilde{\mathbb{L}}_{\{\alpha,\beta\}} + 
\tilde{\mathbb{L}}_{\{\alpha,\beta\}}^t \right)
+i\left(  \tilde{\mathbb{L}}_{\{\alpha,\beta\}}
 -  \tilde{\mathbb{L}}_{\{\alpha,\beta\}}^t 
\right)
= \left( \begin{array}{cc}
K_{11} & K_{12}\\
K_{12}^{\ast} & K_{22}
\end{array} \right)
\end{equation}
to be positive semi-definite. Since the diagonal entries of
$\tilde{\mathbb{K}}_{\{\alpha,\beta \}}$ are obviously
non-negative, this condition reduces to
$ {{\rm Det}} \tilde{\mathbb{K}}_{\{ \alpha,\beta\} }
= K_{11}K_{22} - \abs{K_{12}}^2 \geq 0.$ 
Finally, expressing the $K_{ij}$ again in terms of the $L_{ij}$ 
yields the new constraint
\begin{equation}\label{New Bound 2}
4 L_{11}L_{22} - (L_{12} +L_{21})^2 \geq
3\left(L_{12}- L_{21} \right)^2.
\end{equation}
This bound that holds for the elements of any $2\times 2$ principal
submatrix of the 
full matrix of kinetic coefficients $\mathbb{L}$, irrespective of
the number $n$ of terminals is our second main result. 
Compared to relation (\ref{New Bound 2}), the second law only
requires
$\tilde{\mathbb{L}}_{\{\alpha,\beta\}}$ to be positive 
semi-definite, which is equivalent to $L_{11}, \; L_{22} \geq 0$
and the weaker constraint
\begin{equation}\label{Second Law 2}
4 L_{11}L_{22} - \left(L_{12} +L_{21} \right)^2 \geq 0.
\end{equation}
Note that the reciprocal relations 
(\ref{Time Reversal Symmetry for L})
do not lead to any further relations between the kinetic
coefficients contained in $\tilde{\mathbb{L}}_{\{\alpha,\beta\}}$
for a fixed magnetic field $\mathbf{B}$.

At this point, we emphasize that the procedure shown
here for $2\times 2$ principal submatrices of $\mathbb{L}$
could be easily extended to larger principal submatrices. 
The result would be a whole hierachy of constraints involving
more and more kinetic coefficients.
However, (\ref{New Bound 2}) is the strongest bound following
from (\ref{Asymmetry Bound}), which can expressed in terms 
of only four of these coefficients.

\section{Bounds on Efficiencies}

In this section, we explore the consequences of the bound 
(\ref{Asymmetry Bound}) on the performance of various 
thermoelectric devices.

\subsection{Heat engine}

A thermoelectric heat engine uses heat from a hot reservoir 
as input and generates power output by driving a particle current 
against an external field or a gradient of chemical potential
\cite{Humphrey2005a}. Such an engine can be realized within
the multi-terminal model by considering the terminals $3,\dots, n$
as pure probe terminals, which mimic inelastic scattering events 
while not contributing to the actual transport process. This constraint
reads
\begin{equation}\label{Probe Terminal Constraints}
\vecttt{0}{\vdots}{0}= \vecttt{\mathbf{J}_3}{\vdots}{\mathbf{J}_n}
= \left(\! \begin{array}{ccc}
\mathbb{L}_{32}  & \cdots & \mathbb{L}_{3n}\\
\vdots           & \ddots & \vdots\\
\mathbb{L}_{n2}  & \cdots & \mathbb{L}_{nn}
\end{array} \! \right) 
\vecttt{\aff{F}_2}{\vdots}{\aff{F}_n}.
\end{equation}
By assuming the matrix
\begin{equation}\label{HE Inverted Matrix}
\mathbb{L}_{\{3,\dots,n\}} \equiv \left(\! \begin{array}{ccc}
\mathbb{L}_{33}  & \cdots & \mathbb{L}_{3n}\\
\vdots           & \ddots & \vdots\\
\mathbb{L}_{n3}  & \cdots & \mathbb{L}_{nn}
\end{array} \! \right) 
\end{equation}
to be invertible, we can solve
the self-consistency relations (\ref{Probe Terminal Constraints})
for $\aff{F}_3,\dots,\aff{F}_n$ obtaining
\begin{equation}
\vecttt{\aff{F}_3}{\vdots}{\aff{F}_n} 
= -\left( \mathbb{L}_{\{3,\dots,n\}} \right)^{-1} 
\vecttt{\mathbb{L}_{32}}{\vdots}{\mathbb{L}_{n3}}\aff{F}_2.
\end{equation}
After inserting this solution into 
(\ref{Phenomenological Equations})
and identifying the heat current $J_q\equiv J_2^q$ leaving the 
hot reservoir and the particle current $J_{\rho}\equiv J_2^{\rho}$,
we end up with the reduced system
\begin{equation}
\vectt{J_{\rho}}{J_q}
= \mathbb{L}^{{{\rm HE}}}
\vectt{\mathcal{F}_{\rho}}{\mathcal{F}_q},
\end{equation}
of phenomenological equations. Here, the effective matrix of 
kinetic coefficients  is given by 
\begin{equation}\label{HE Onsager Matrix}
\mathbb{L}^{{{\rm HE}}} \equiv \mathbb{L}_{22} -
\left(\mathbb{L}_{23}, \cdots, \mathbb{L}_{2n} \right)
\left(\mathbb{L}_{\{3,\dots,n\}}\right)^{-1}
\vecttt{\mathbb{L}_{32}}{\vdots}{\mathbb{L}_{n3}}
\equiv \left( \! \begin{array}{cc}
L_{\rho \rho} & L_{ \rho q}\\
L_{q \rho}    & L_{qq}
\end{array} \! \right)
\end{equation}
and the affinities $\mathcal{F}_{\rho}\equiv\mathcal{F}_2^{\rho}=
\Delta\mu_2/T<0$ and $\mathcal{F}_q\equiv\mathcal{F}_2^q=\Delta 
T_2/T^2>0$ have to be chosen such that $J_{\rho},J_{q}\geq 0$ for
the model to work as a proper heat engine.\\

$\mathbb{L}^{{{\rm HE}}}$ is not a principal submatrix of the full 
Onsager matrix $\mathbb{L}$ and therefore the bound 
(\ref{Asymmetry Bound}) does not apply directly. However, 
$\mathbb{L}^{{{\rm HE}}}$ can be written as the Schur complement $\mathbb{L}/\mathbb{L}_{\{3,\dots,n\}}$ (see
\ref{Apx AI Schur complements} for the definition), the asymmetry
index of which is dominated
by the asymmetry index of $\mathbb{L}$ as proven in Proposition 
\ref{Apx Dominance of Schur Complement} of 
\ref{Apx AI Schur complements}. Consequently, we have 
\begin{equation}
\asym{\mathbb{L}^{{{\rm HE}}}} = 
\asym{\mathbb{L}/\mathbb{L}_{\{3,\dots,n\}}} \leq
\asym{\mathbb{L}} \leq \cot  \left( \frac{\pi}{n}\right).
\end{equation}
or, equivalently,
\begin{equation}\label{HE Bound on Onsager Coefficients}
4L_{\rho\rho}L_{qq} 
-\left(L_{\rho q}+L_{q \rho}\right)^2
\geq \tan^2 \left(\frac{\pi}{n}\right)
\left(L_{\rho q}-L_{q\rho}\right)^2.
\end{equation}
This constraint
shows that whenever $L_{\rho q}\neq L_{q\rho}$, 
the entropy production (\ref{Entropy Production}) must be
strictly larger than zero, thus ruling out the option of 
dissipationless transport generated solely by reversible 
currents for any model with a finite number $n$ of terminals.
For any $n>3$ this constraint is weaker than (\ref{New Bound 2}).
The reason is that the Onsager coefficients in 
(\ref{HE Bound on Onsager Coefficients}) are not elements of 
the full matrix (\ref{Onsager Matrix}) but rather involve 
the inversion of $\mathbb{L}_{\{3,\dots,n\}}$ defined in 
(\ref{HE Inverted Matrix}). 
Still, this constraint is stronger than the bare second law,
which requires only
\begin{equation}\label{HE Second Law}
4L_{\rho\rho}L_{qq}-\left(L_{\rho q}+L_{q \rho}\right)^2\geq 0,
\end{equation}
irrespective of whether or not $\mathbb{L}^{{{\rm HE}}}$ is 
symmetric.\\

\begin{figure}[t]
\centering
\epsfig{file=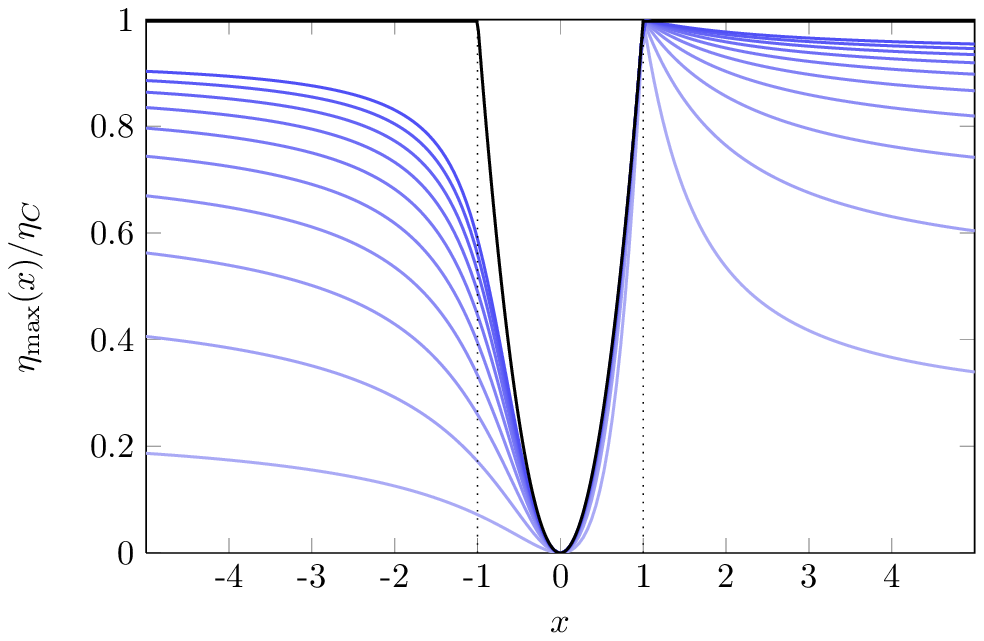, scale=1.15}
\epsfig{file=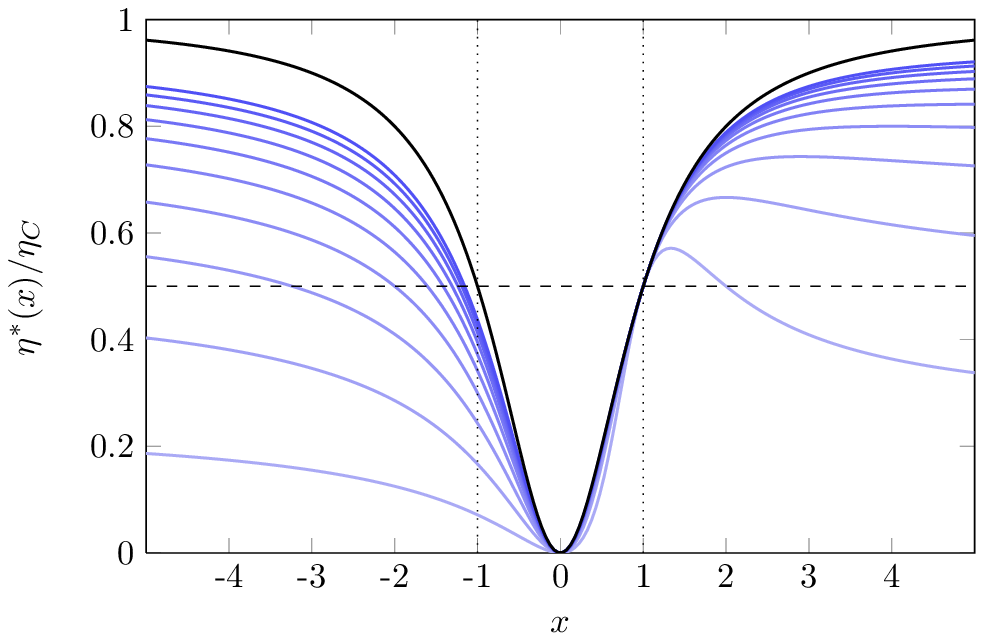, scale=1.15}
\caption{Bounds on the efficiency of the multi-terminal model as 
a thermoelectric heat engine as functions of the asymmetry 
parameter $x$ and in units of $\eta_C$. The upper panel shows 
$\eta_{{{\rm max}}}(x)$  (see (\ref{HE ME as function of x})),
the lower one $\eta^{\ast}(x)$ (see
(\ref{HE EMP as function of x})). In both panels, the blue lines, from
bottom to top, belong to models with $n=3,\dots,12$ terminals and
the solid, black line corresponds to the bound following from the
bare second law as obtained by Benenti \etal \cite{Benenti2011}.
The dashed line  in the lower panel 
marks the Curzon-Ahlborn limit $\eta_{CA}=\eta_C/2$.
\label{Fig HE Efficiecy}
}
\end{figure}

The constraint (\ref{HE Bound on Onsager Coefficients}) implies 
a constraint on the efficiency $\eta$ of such a  particle-exchange 
heat engine \cite{Humphrey2005a}, which is defined as 
\begin{equation}
\eta \equiv - \frac{\Delta\mu_2 J_{\rho}}{J_q}\leq\eta_C.
\end{equation}
Like for any heat engine, this efficiency is subject to the 
Carnot-bound $\eta_C\equiv 1-T/T_2$, which, in the linear response
regime, is given by $\eta_C\approx \Delta T_2/T = T\mathcal{F}_q$.
Following Benenti \etal \cite{Benenti2011}, we now introduce
the dimensionless parameters 
\begin{equation}\label{HE Dimless Parameters}
y\equiv\frac{L_{\rho q}L_{q\rho}}
{L_{\rho\rho}L_{qq}-L_{\rho q}L_{q\rho}}
\qquad {{{\rm and}}} \qquad
x \equiv \frac{L_{\rho q}}{L_{q\rho}},
\end{equation}
which allow us to write the maximum efficiency of the engine
$\eta_{{{\rm max}}}$  (under the condition $J_q>0$) in the instructive 
form 
\cite{Benenti2011}
\begin{equation}\label{HE Efficiency in xy}
\eta_{{{\rm max}}}(x,y) = \eta_C x \frac{\sqrt{y+1}-1}{\sqrt{y+1}+1}.
\end{equation}
Restating the new bound (\ref{HE Bound on Onsager Coefficients}) in
terms of $x$ and $y$ yields
\begin{equation}\label{HE y Constraint}
\begin{array}{lll}
h_n(x)\leq y \leq 0 & \quad {{\rm if}}\quad & x<0,\\
0     \leq y \leq h_n(x) & \quad {{\rm if}}\quad & x>0
\end{array}
\end{equation}
with 
\begin{equation}\label{HE h Function}
h_n(x) \equiv \frac{4x}{(x-1)^2}\cos^2\left( \frac{\pi}{n}\right).
\end{equation}
Consequently, maximizing (\ref{HE Efficiency in xy}) with respect 
to $y$ yields the optimal $y^{\ast}(x)=h_n(x)$ and the maximum efficiency 
\begin{equation}\label{HE ME as function of x}
\eta_{{{\rm max}}}(x)\equiv  \eta_{{{\rm max}}}(x,y^{\ast}(x)) = 
\eta_C x \frac{\sqrt{4x\cos^2 (\pi/n) +(x-1)^2}-|x-1|}
{\sqrt{4x\cos^2 (\pi/n) +(x-1)^2}+ |x-1|}.
\end{equation} 
This bound is plotted in figure \ref{Fig HE Efficiecy}
as a function of $x$ for an increasing number $n$ of 
terminals. For $n=3$, we recover the result obtained in our 
preceding work on the three terminal model \cite{Brandner2013}.
In the limit $n\rightarrow\infty$, $\eta_{{{\rm max}}}(x)$ 
converges to the bound derived by Benenti \etal
\cite{Benenti2011} within a general analysis relying only on the
second law. However, for any finite $n$, $\eta_{{{\rm max}}}(x)$
is constrained to be strictly smaller than $\eta_C$, as soon as
$x$ deviates from $1$. Thus, from the perspective of maximum
efficiency, breaking the time reversal symmetry is not
beneficial.\\

As a second important benchmark for the performance of a heat
engine, we consider its efficiency at maximum power 
$\eta^{\ast}$ \cite{Curzon1975, Esposito2009, Seifert2012} 
obtained by maximizing the power output 
\begin{equation}
P_{{{\rm out}}}\equiv -\Delta \mu_2 J_{\rho}
=- T \mathcal{F}_{\rho} \left( L_{\rho\rho}\mathcal{F}_{\rho}
+L_{\rho q} \mathcal{F}_q \right)
\end{equation}
with respect to $\mathcal{F}_{\rho}$ for fixed $\mathcal{F}_q$.
In terms of the dimensionless parameters (\ref{HE Dimless Parameters}), 
it reads \cite{Benenti2011}
\begin{equation}
\eta^{\ast}(x,y) = \eta_C \frac{xy}{4+2y}
\end{equation}
and attains its maximum 
\begin{equation}\label{HE EMP as function of x}
\eta^{\ast}(x)\equiv \eta^{\ast}(x,y^{\ast}(x)) = \eta_C 
\frac{x^2 \cos^2(\pi/n)}{(x-1)^2 + 2x \cos^2(\pi/n)}
\end{equation}
at $y^{\ast}(x)=h_n(x)$. In the lower panel of figure \ref{Fig HE Efficiecy},
$\eta^{\ast}(x)$ is plotted as a function of the
asymmetry parameter $x$. For $x=1$, this bound acquires
the Curzon-Ahlborn value $\eta_{CA}\equiv
\eta_C/2$. For $x\neq 1$, however it can become significantly higher
even for a small number $n$ of terminals. Specifically, we observe 
that $\eta^{\ast}(x)$ exceeds $\eta_{CA}$ for any $n\geq 3$ in a 
certain range of $x$ values. For $n\geq 4$, this range includes all $x>1$.
Furthermore,
$\eta^{\ast}(x)$ attains its global maximum 
\begin{equation}
\hat{\eta}^{\ast\ast}\equiv
 \frac{\eta_C}{1+ \sin^2\left(\frac{\pi}{n}\right)}
\end{equation}
at the finite value 
$ x = 1/\sin^2 \left(\frac{\pi}{n}\right)$.
Remarkably, both  $\eta_{{{\rm max}}}(x)$ and $\eta^{\ast}
(x)$ approach the same asymptotic value $\eta^{\infty} 
\equiv \eta_C \cos^{2} \left(\frac{\pi}{n}\right)$
for $x\rightarrow\pm\infty$.

\subsection{Refrigerator}

In the preceding section, we discussed the performance of the 
multi-terminal model if it is operated as a heat engine. 
Quite naturally, we can change the mode of operation of 
this engine such that it functions as a refrigerator. The 
resulting device consumes electrical power from which it generates
a heat current from the cold to the hot reservoir. Thus, compared
to the heat engine, input and output are interchanged and the 
affinities $\mathcal{F}_{\rho}<0$ and $\mathcal{F}_q>0$ have to be 
chosen such that both currents $J_{\rho}$ and $J_q$ are negative.\\

\begin{figure}[t]
\centering 
\epsfig{file=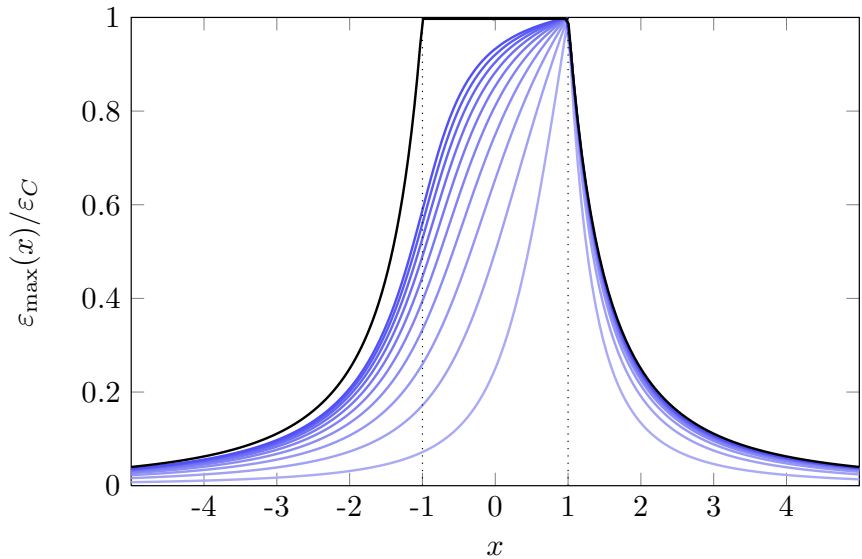, scale=1.15}
\caption{Maximum coefficient of performance $\varepsilon_{{{\rm 
max}}}(x)$ (see (\ref{R ME as function of x})) of a thermoelectric 
refrigerator as a function of the asymmetry parameter $x$. The 
blue lines from bottom to top represent models with $n=3,\dots,12$ 
terminals. The black curve shows the bound required by the bare
second law, which is asymptotically reached in the limit 
$n\rightarrow\infty$. \label{Fig Max COP Refrigerator}}
\end{figure}

Analogously to the case of the heat engine, we will now show that
the bound (\ref{HE Bound on Onsager Coefficients}) on the 
kinetic coefficients constrains the performance of the 
thermoelectric refrigerator described above. To this end, we will
use the coefficient of performance \cite{Callen1985}
\begin{equation}\label{R Efficiency}
\varepsilon\equiv 
-\frac{J_q}{\Delta \mu_2 J_{\rho}}.
\end{equation}
as a benchmark parameter. Its upper bound following from the second
law is given by $\varepsilon_C\equiv T/\Delta T_2 = 1/(T\mathcal{F}_q)$, 
which is the efficiency of the ideal refrigerator.  
In this sense, $\varepsilon_C$ is the analogue to the Carnot 
efficiency.\\

Taking the maximum of $\varepsilon$ over $\mathcal{F}_{\rho}$
(under the condition $J_{\rho}<0$) while keeping $\mathcal{F}_q$
fixed, yields the maximum coefficient of performance
\cite{Benenti2011}
\begin{equation}
\varepsilon_{{{\rm max}}}(x,y) = \frac{\varepsilon_C}{x}
\frac{\sqrt{y+1}-1}{\sqrt{y+1}+1}.
\end{equation}
Here, we used again the dimensionless parameters defined in 
(\ref{HE Dimless Parameters}). Since $y$ is subject to the 
constraint (\ref{HE y Constraint}), $\varepsilon_{{{\rm 
max}}}(x,y)$ attains its maximum
\begin{equation}\label{R ME as function of x}
\varepsilon_{{{\rm max}}}(x) \equiv  \varepsilon_{{{\rm max}}}
(x,y^{\ast}(x)) = \frac{\varepsilon_C}{x}
\frac{\sqrt{4x\cos^2 (\pi/n) +(x-1)^2}-|x-1|}
{\sqrt{4x\cos^2 (\pi/n) +(x-1)^2}+|x-1|}
\end{equation}
with respect to $y$ at $y^{\ast}(x)=h_n(x)$,
where $h_n(x)$ was introduced in (\ref{HE h Function}). Figure 
\ref{Fig Max COP Refrigerator} shows $\eta_{{{\rm max}}}(x)$ 
for models with an increasing number of probe terminals $n$.
For any finite $n$, $\varepsilon_C$ can only be reached for the
symmetric value $x=1$. The black line follows solely from the
second law (\ref{HE Second Law}) and would in principle allow to 
reach $\varepsilon_C$ with finite current for $x$ between 
$-1$ and $1$. However, like for the heat engine, our analysis 
reveals that such a high performance refrigerator would need to 
be equipped with an infinite number of terminals.

\subsection{Isothermal Engine}

By an isothermal, thermoelectric engine, we understand in this 
context a device in which one particle current driven by a
(negative) gradient in chemical potential drives another one 
uphill a chemical potential gradient at constant temperature 
$T$. In order to implement such a machine within the 
multi-terminal framework, we put $\mathcal{F}_2^q =
\cdots = \mathcal{F}_n^q=0$. The remaining affinities 
$\mathcal{F}_2^{\rho},\dots,\mathcal{F}_n^{\rho}$ are connected
to the particle currents via a reduced set of phenomenological 
equations given by 
\begin{equation}\label{IT Unreduced Phenomenologica Equations}
\vecttt{J^{\rho}_2}{\vdots}{J^{\rho}_n} = \left( \!
\begin{array}{ccc}
(\mathbb{L}_{22})_{11} & \cdots & (\mathbb{L}_{2n})_{11}\\
\vdots                 & \ddots & \vdots\\
(\mathbb{L}_{n2})_{11} & \cdots & (\mathbb{L}_{nn})_{11} 
\end{array} \! \right) 
\vecttt{\mathcal{F}^{\rho}_2}{\vdots}{\mathcal{F}^{\rho}_n},
\end{equation}
where $(\mathbb{L}_{\alpha\beta})_{11}$ denotes the (11)-entry of
the block matrix $\mathbb{L}_{\alpha\beta}$ defined in 
(\ref{Landauer Buttiker}). We note that the heat currents $J^q_2,
\dots, J^q_n$ do not necessarily have to vanish. However, since
they do not contribute to the entropy production
(\ref{Entropy Production}), they are irrelevant in the present
analysis. Similar to the treatment of the heat engine, 
we put $J^{\rho}_4= \cdots = J^{\rho}_n =0$, thus considering the
terminals $4,\dots,n$ as pure probe terminals simulating 
inelastic scattering events. Consequently,
(\ref{IT Unreduced Phenomenologica Equations}) can be reduced 
further to the generic form
\begin{equation}
\vectt{J_2^{\rho}}{J_3^{\rho}} = 
\mathbb{L}^{{{\rm IE}}}
\vectt{\mathcal{F}_2^{\rho}}{\mathcal{F}_3^{\rho}}.
\end{equation}
Here, we have introduced the matrix 
\begin{eqnarray}
\mathbb{L}^{{{\rm IE}}} & \equiv 
\left. \left( \!
\begin{array}{ccc}
(\mathbb{L}_{22})_{11} & \cdots & (\mathbb{L}_{2n})_{11}\\
\vdots                 & \ddots & \vdots\\
(\mathbb{L}_{n2})_{11} & \cdots & (\mathbb{L}_{nn})_{11} 
\end{array} \! \right) \right/
\left( \!
\begin{array}{ccc}
(\mathbb{L}_{44})_{11} & \cdots & (\mathbb{L}_{4n})_{11}\\
\vdots                 & \ddots & \vdots\\
(\mathbb{L}_{n4})_{11} & \cdots & (\mathbb{L}_{nn})_{11} 
\end{array} \! \right) \nonumber \\
& \equiv\left( \! 
\begin{array}{cc}
L_{22} & L_{23}\\
L_{32} & L_{33}
\end{array} \! \right)
\end{eqnarray}
again using the Schur complement defined in \ref{Apx AI Schur complements}.
The affinities $\mathcal{F}^{\rho}_2, \mathcal{F}^{\rho}_3>0$
have to be chosen such that $J^{\rho}_2$ is negative and $J^{\rho}_3$
is positive to ensure that the device pumps particles into the 
reservoir $2$ against the gradient in chemical potential $\Delta
\mu_2$. \\

We will now derive a bound on the elements of $\mathbb{L}^{{{
\rm IE}}}$. By employing expression
(\ref{MT Landauer Buttiker Tensor Product}), we can write
\begin{eqnarray}
\left( \!
\begin{array}{ccc}
(\mathbb{L}_{22})_{11} & \cdots & (\mathbb{L}_{2n})_{11}\\
\vdots                 & \ddots & \vdots\\
(\mathbb{L}_{n2})_{11} & \cdots & (\mathbb{L}_{nn})_{11} 
\end{array} \! \right)
& = \frac{T e^2}{h}\int_{-\infty}^{\infty} dE \; F(E) 
\left( \id - \bar{\mathbb{T}}(E) \right)
\nonumber \\
& \equiv \mathcal{N}
 \left( \id - \langle \bar{\mathbb{T}}\rangle \right)
\end{eqnarray}
with
\begin{equation}
\mathcal{N}
\equiv \frac{Te^2}{h}\int_{-\infty}^{\infty} dE \; F(E)
= \frac{Te^2}{h}
\end{equation}
and 
\begin{equation}
\langle \bar{\mathbb{T}}\rangle \equiv 
\int_{-\infty}^{\infty} dE \; F(E)
\bar{\mathbb{T}}(E).
\end{equation}
Since $\bar{\mathbb{T}}(E)$ is doubly substochastic for any $E$,
the matrix $\langle \bar{\mathbb{T}}\rangle$ is also
doubly substochastic. Therefore, by applying Corollary 
\ref{Apx AI SC of Doubls SS Matrices} of \ref{Apx AI Schur complements},
we find 
\begin{eqnarray}\label{IE AI}
\asym{\mathbb{L}^{{{\rm IE}}}}
=\asym{\frac{\mathbb{L}^{{{\rm IE}}}}{\mathcal{N}}}
& =\asym{\left.\left(\id-\langle\bar{\mathbb{T}}\rangle
\right)\right/
\left(\id -\langle\bar{\mathbb{T}}\rangle\right)_{\{3,\dots,n-1\}}}
\nonumber\\
& \leq \cot\left(\frac{\pi}{3}\right) = \frac{1}{\sqrt{3}},
\end{eqnarray}
where $\left(\id-\langle\bar{\mathbb{T}}\rangle\right)_{
\{3,\dots,n-1\}}$ denotes the principal submatrix of 
$\id-\langle\bar{\mathbb{T}}\rangle$ consisting of all
but the first two rows and columns. 
Expressing (\ref{IE AI}) in terms of the elements of 
$\mathbb{L}^{{{\rm IE}}}$ gives the bound 
\begin{equation}\label{IE Bond on Onsager Coefficients}
4 L_{22}L_{33}- \left(
L_{23}+L_{32}\right)^2 \leq 3\left(L_{23}-L_{32}\right)^2.
\end{equation}
We emphasize that, in contrast to the bound 
(\ref{HE Bound on Onsager Coefficients}) we derived for the heat
engine, the bound (\ref{IE Bond on Onsager Coefficients}) is independent 
of the number of  probe terminals involved in the device.\\

\begin{figure}[t]
\centering
\epsfig{file=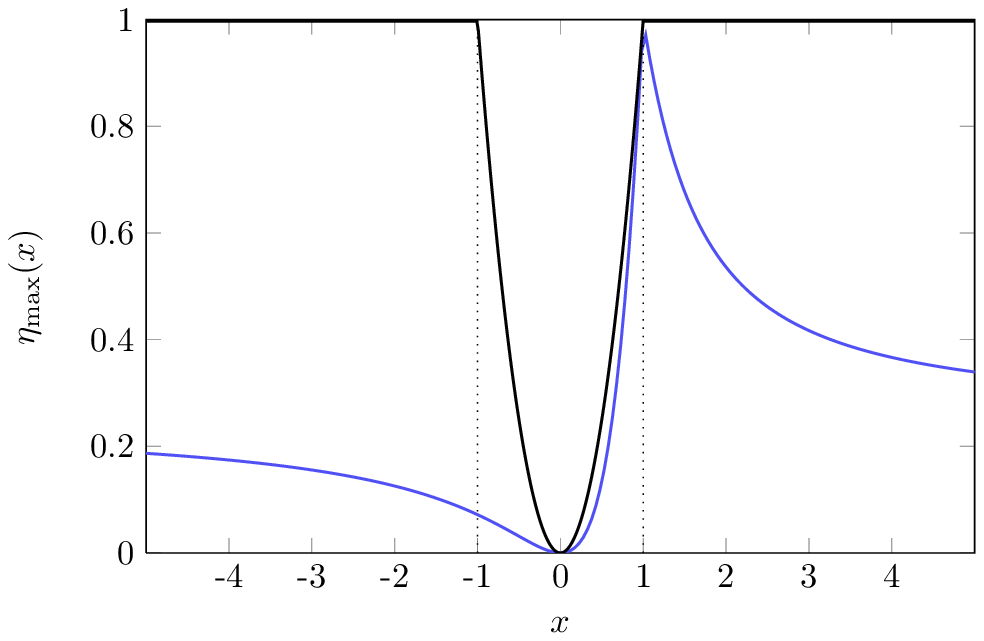, scale=0.78}
\epsfig{file=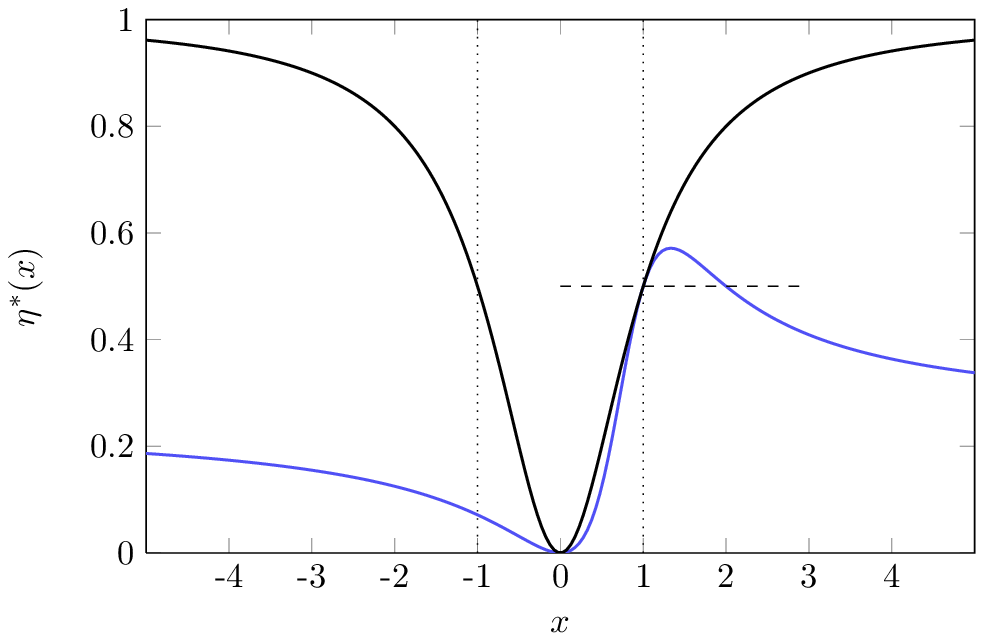, scale=0.78}
\caption{Bounds on benchmark parameters for the performance 
of the isothermal, thermoelectric engine as functions of the 
asymmetry parameter $x$. The right panel shows the maximum 
efficiency $\eta_{{{\rm max}}}(x)$ (see
(\ref{IE ME as function of x})), the left one efficiency at maximum
power $\eta^{\ast}(x)$ (see (\ref{IE EMP as function of x})). 
The black lines follow from the bare second law, the blue lines from
the stronger constraint (\ref{IE Bond on Onsager Coefficients}). 
Both, $\eta_{{{\rm max}}}(x)$ and $\eta^{\ast}(x)$
asymptotically reach the value $1/4$. The 
dashed line in the right plot marks the value $1/2$ of $\eta^{
\ast}(x)$ at the symmetric value $x=1$.
\label{Fig IE Performance}}
\end{figure}

In the next step we explore the implications of
(\ref{IE Bond on Onsager Coefficients}) for the performance of 
the isothermal engine. To this end, we identify the output power
of the device as 
\begin{equation}
P_{{{\rm out}}}\equiv -\Delta \mu_2 J^{\rho}_2 
= - T \mathcal{F}^{\rho}_2 J^{\rho}_2
\end{equation}
and correspondingly the input power as 
\begin{equation}
P_{{{\rm in}}}\equiv \Delta \mu_3 J^{\rho}_3 
= T \mathcal{F}^{\rho}_3 J^{\rho}_3. 
\end{equation}
Consequently, the efficiency of the isothermal engine
reads
\begin{equation}
\eta = \frac{P_{{{\rm out}}}}{P_{{{\rm in}}}} = - 
\frac{ \mathcal{F}^{\rho}_2 J^{\rho}_2}
{\mathcal{F}^{\rho}_3 J^{\rho}_3}.
\end{equation}
We note that, in the situation considered here, the entropy 
production (\ref{Entropy Production}) reduces to 
\begin{equation}
\dot{S}= \mathcal{F}^{\rho}_2 J^{\rho}_2
+\mathcal{F}^{\rho}_3 J^{\rho}_3
\end{equation}
and thus the second law $\dot{S}\geq 0$ requires $\eta\leq 1$
for isothermal engines \cite{Seifert2012}.\\

Optimizing $\eta$ and  $P_{{{\rm out}}}$ (under the condition 
$J_3^{\rho} >0$) with respect to $\mathcal{F}_2^{\rho}$
while keeping $\mathcal{F}_3^{\rho}$ fixed yields the 
maximum efficiency 
\begin{equation}
\eta_{{{\rm max}}}(x,y)= x \frac{\sqrt{y+1}-1}{\sqrt{y+1}+1}
\end{equation}
and the efficiency at maximum power 
\begin{equation}
\eta^{\ast}(x,y) = \frac{xy}{4+2y},
\end{equation}
where we have introduced the dimensionless parameters 
\begin{equation}
y \equiv \frac{L_{23}L_{32}}{L_{22}L_{33}-L_{23}L_{32}}
\qquad {{\rm and}} \qquad
x \equiv \frac{L_{23}}{L_{32}}
\end{equation}
analogous to (\ref{HE Dimless Parameters}). Using these 
definitions, the bound (\ref{IE Bond on Onsager Coefficients})
translates to 
\begin{equation}
\begin{array}{lll}
h(x)\leq y \leq 0 \qquad &{{\rm if}} \qquad & x<0,\\
0\leq y \leq h(x) \qquad &{{\rm if}} \qquad & x>0
\end{array}
\end{equation}
with
\begin{equation}
h(x)= \frac{x}{(x-1)^2}
\end{equation}
and $\eta_{{{\rm max}}}(x,y)$ as well as $\eta^{\ast}(x,y)$ attain
their respective maxima with respect to $y$ at $y^{\ast}=h(x)$. The resulting 
bounds 
\begin{equation}\label{IE ME as function of x}
\eta_{{{\rm max}}}(x)\equiv \eta_{{{\rm max}}}(x,y^{\ast}(x))
=x \frac{\sqrt{x^2-x+1}-|x-1|}{\sqrt{x^2-x+1}+|x-1|}
\end{equation}
and 
\begin{equation}\label{IE EMP as function of x}
\eta^{\ast}(x)\equiv \eta^{\ast}(x,y^{\ast}(x))
=\frac{x^2}{4x^2-6x+4}
\end{equation}
are plotted in figure (\ref{Fig IE Performance}). We observe that
the $\eta_{{{\rm max}}}(x)$ reaches $1$ only for $x=1$ and 
decreases rapidly as the asymmetry parameter $x$ deviates from $1$, 
while  $\eta^{\ast}(x)$ exceeds the Curzon-Ahlborn value
$1/2$ for $x$ between $1$ and $2$ with a global maximum $\eta^{\ast\ast}=
4/7$ at $x=4/3$.  In contrast to the non-isothermal engines analyzed 
in the preceding sections, all these bounds do not depend on the number
of probe terminals. 

\subsection{Absorption Refrigerator}

By an absorption refrigerator, one commonly understands a device
that generates a heat current cooling a hot reservoir, while itself
being supplied by a heat source \cite{Palao2001,Skrzypczyk2011}. 
The multi-terminal model allows to implement such a device by 
following a very similar strategy like the one used for the 
isothermal engine, i.e., we put $\mathcal{F}_2^{\rho}=\dots
=\mathcal{F}_n^{\rho}=0$ and end up with the reduced system of 
phenomenological equations
\begin{equation}\label{AR Full PhenEq}
\vecttt{J_2^q}{\vdots}{J_n^q}=
\left(\! \begin{array}{ccc}
(\mathbb{L}_{22})_{22} & \cdots & (\mathbb{L}_{2n})_{22}\\
\vdots                 & \ddots & \vdots\\
(\mathbb{L}_{2n})_{22} & \cdots & (\mathbb{L}_{nn})_{22}
\end{array}\! \right)
\vecttt{\mathcal{F}_2^q}{\vdots}{\mathcal{F}_n^q}
\end{equation}
connecting the heat currents with the temperature gradients. 
Assuming the terminals $4,\dots,n$ to be pure probe terminals
then leads to
\begin{equation}
\vectt{J_2^q}{J_3^q} = 
\mathbb{L}^{{{\rm AR}}}
\vectt{\mathcal{F}^q_2}{\mathcal{F}_3^q},
\end{equation}
where $\mathcal{F}_2^q<0$, $\mathcal{F}_3^q>0$ have to be 
adjusted such that $J^2_q>0$ and $J^3_q>0$. The matrix 
$\mathbb{L}^{\rm AR}$ is given by
\begin{eqnarray}
\mathbb{L}^{{{\rm AR}}}
& = \left. \left(\! \begin{array}{ccc}
(\mathbb{L}_{22})_{22} & \cdots & (\mathbb{L}_{2n})_{22}\\
\vdots                 & \ddots & \vdots\\
(\mathbb{L}_{2n})_{22} & \cdots & (\mathbb{L}_{nn})_{22}
\end{array}\! \right)\right/
\left(\! \begin{array}{ccc}
(\mathbb{L}_{44})_{22} & \cdots & (\mathbb{L}_{4n})_{22}\\
\vdots                 & \ddots & \vdots\\
(\mathbb{L}_{4n})_{22} & \cdots & (\mathbb{L}_{nn})_{22}
\end{array}\! \right) \nonumber\\
& \equiv \left( \! \begin{array}{cc}
L_{22}' & L_{23}'\\
L_{32}' & L_{33}'
\end{array}\! \right)
\end{eqnarray}
and by following the reasoning of the last section, we can derive 
the bound
\begin{equation}\label{AR Bond on Onsager Coefficients}
4 L_{22}'L_{33}'- \left(
L_{23}'+L_{32}'\right)^2
\leq 3\left(L_{23}'-L_{32}'\right)^2.
\end{equation}

\begin{figure}[t]
\centering 
\epsfig{file=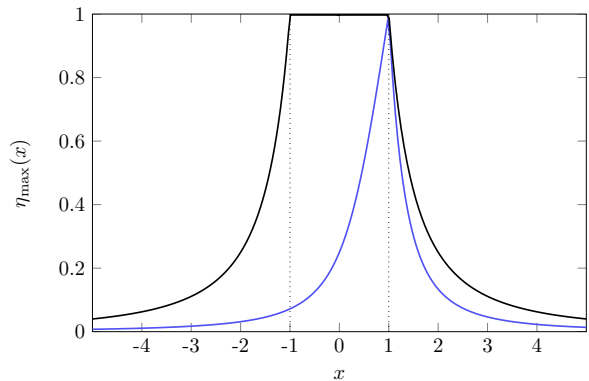, scale=0.78}
\caption{Maximum efficiency $\eta_{{{\rm max}}}(x)$ (see 
(\ref{AR Efficiency Bound})) of the thermoelectric absorption 
refrigerator as a function of $x$. The blue line follows by virtue
of the constraint (\ref{AR Bond on Onsager Coefficients}), the black line
by invoking only the second law. 
 \label{Fig Max Efficiency AR}}
\end{figure}

The efficiency of the absorption refrigerator can be consistently 
defined as 
\begin{equation}\label{AR Efficiency}
\eta\equiv - \frac{\Delta T_2 J_2^q}{\Delta T_3 J^q_3}
= - \frac{\mathcal{F}^q_2 J^q_2}{\mathcal{F}^q_3 J^q_3}\leq 1. 
\end{equation}
Just like for the isothermal engine, after maximizing this efficiency 
over $\mathcal{F}_2^q$ (under the condition $J_2^q>0$), we can derive 
an upper bound 
\begin{equation}\label{AR Efficiency Bound}
\eta_{{{\rm max}}}(x)\equiv 
\frac{1}{x}\frac{\sqrt{x^2-x+1}-|x-1|}{\sqrt{x^2-x+1}+|x-1|}
\end{equation}
from (\ref{AR Bond on Onsager Coefficients}). Again, this bound is 
independent of the number of probe terminals. Figure 
\ref{Fig Max Efficiency AR} shows it as a function of the asymmetry 
parameter $x \equiv L_{23}'/L_{32}'$.\\

For completeness, we emphasize that the efficiency (\ref{AR Efficiency})
used here differs from the coefficient of performance 
\begin{equation}
\varepsilon\equiv \frac{J^q_2}{J^q_3} =
\frac{L'_{22}\mathcal{F}^q_2+L'_{23}\mathcal{F}_3^q}{
L'_{32}\mathcal{F}^q_2 + L'_{33}\mathcal{F}^q_{3}}
\end{equation}
used as a benchmark parameter in \cite{Palao2001} and
\cite{Skrzypczyk2011}. Since $\varepsilon$ is 
unbounded in the linear response regime, maximization 
with respect to $\mathcal{F}^q_2$ or $\mathcal{F}^q_3$ would be
meaningless. 

\section{Conclusion and Outlook}

We have studied the influence of broken time reversal symmetry on 
thermoelectric transport within the quite general framework of an
$n$-terminal model. Our analytical calculations prove that the 
asymmetry index of any principal submatrix of the full Onsager
matrix  defined in (\ref{Onsager Matrix}) is bounded according to
(\ref{Asymmetry Bound}). This somewhat abstract bound can be 
translated into the set (\ref{New Bound 2}) of new constraints on
the kinetic coefficients. Any of these 
constraints is obviously stronger than the bare second law and can
not be deduced from Onsagers time reversal argument. Furthermore,
we note that it is straight forward to repeat the procedure carried
out in section 3.2 for larger principal submatrices, thus obtaining
relations analogous to (\ref{New Bound 2}), which involve successively
higher order products of kinetic coefficients. Investigating this 
hierarchy of constraints will be left to future work. \\

After the general analysis of the transport processes in the full 
multi-terminal set-up, we investigated the consequences of our new
bounds on the performance of the model if operated as a thermoelectric
heat engine. We found that both the maximum efficiency as well as the
efficiency at maximum power are subject to bounds, which strongly depend 
on the number $n$ of terminals. In the minimal case $n=3$, we recover the 
strong bounds already discussed in \cite{Brandner2013}. Although our new 
bounds become successively weaker as $n$ is increased, they prove that 
reversible transport is impossible in any situation with a finite number
of terminals. Only in the limit $n\rightarrow\infty$ we are back at the
situation discussed by Benenti \etal \cite{Benenti2011}, in which the
second law effectively is the only constraint. We recall that for $n=3$
our bounds can indeed be saturated as  Balachandran \etal 
\cite{Balachandran2013} have shown within a specific model. Whether or
not it is possible to saturate the bounds for  higher $n$ remains open
at this stage and constitutes an important question for future 
investigations.\\

Like in the case of the heat engine, the bound on the maximum coefficient
of performance we derived for the thermoelectric refrigerator becomes
weaker as $n$ increases. Interestingly, the situation is quite different
for the isothermal engine and the absorption refrigerator considered in 
the sections 4.3 and 4.4. The bounds on the respective benchmark parameters
equal those of the three-terminal case irrespective of the actual number
of terminals involved. If one assumed that any kind of inelastic scattering
could be simulated by a sufficiently large number of probe terminals, one 
had to conclude that the results shown in figures \ref{Fig IE Performance}
and \ref{Fig Max Efficiency AR} were a universal bound on the efficiency 
of any such device. At least, the results of sections 4.3 and 4.4 suggest
a fundamental difference between transport processes under broken time-reversal symmetry that are driven by only one type of affinities, i.e., either
chemical potential differences or temperature differences, and those, which
are induced by both types of thermodynamic forces.\\

We emphasize that technically all our results ultimately rely on the sum rules 
(\ref{Sum Rule for T}) for the elements of the transmission matrix. These
constraints reflect the fundamental law of current conservation, which 
should be seen as the basic physical principle behind our bounds. Therefore
the validity of these bounds is not limited to the quantum realm. It rather 
extends to any model, quantum or classical, for which the kinetic coefficients
can be expressed in the generic form (\ref{Landauer Buttiker}). Some specific
examples for quantum mechanical models  which fulfil this requirement are discussed in \cite{Balachandran2013} and \cite{Sanchez2011}. A classical model
belonging to this class was recently introduced by Horvat \etal
\cite{Horvat2012}.\\

In summary, we have achieved a fairly complete picture of thermoelectric 
transport under broken time reversal symmetry in systems with non-interacting
particles for which the Onsager coefficients can be expressed in the
Landauer-B\"uttiker form (\ref{Landauer Buttiker}). However, fully interacting
systems, which require to go beyond the single particle picture, are not
covered by our analysis yet. Exploring these systems remains one of the major
challenges for future research.

\ack
We gratefully acknowledge stimulating discussions with K. Saito and support 
of the ESF through the EPSD network.  
\newpage

\appendix

\section{Quantifying the asymmetry of positive semi-definite
matrices}\label{Apx AI Basic properties}

We first recall the definition (\ref{Asymmetry Index}) 
\begin{equation}\label{Apx Asymmetry Index}
\mathcal{S}(\mathbb{A}) \equiv 
\min \left\{ \left. s \in \mathbb{R} \right\vert 
\forall \mathbf{z} \in \mathbb{C}^{m} \; \; \;
\mathbf{z}^{\dagger}
\left(s \left( \mathbb{A} + \mathbb{A}^t \right)
+ i \left( \mathbb{A}
-  \mathbb{A}^t \right) \right)\mathbf{z} \geq 0.
\right\},
\end{equation}
of the asymmetry index of an arbitrary positive semi-definite 
matrix $\mathbb{A}\in\mathbb{R}^{m\times m}$. Below, we list 
some of the basic properties of this quantity, which can be 
inferred directly from its definition. 

\begin{proposition}[Basic properties of the asymmetry index]
For any positive semi-definite $\mathbb{A}\in\mathbb{R}^{m\times 
m}$ and  $\lambda >0$, we have
\begin{equation}
\asym{\mathbb{A}}= \asym{\lambda\mathbb{A}}= \asyms{\mathbb{A}^t}
\end{equation}
and
\begin{equation}
\asym{\mathbb{A}}\geq 0
\end{equation}
with equality if and only if $\mathbb{A}$ is symmetric.
If $\mathbb{A}$ is invertible, it holds additionally
\begin{equation}
\asym{\mathbb{A}}=\asyms{\mathbb{A}^{-1}}.
\end{equation}
\end{proposition}

Furthermore, we can easily prove the following two propositions,
which are crucial for the derivation of our main results. 

\begin{proposition}[Convexity of the asymmetry index]
\label{Apx Convexity AI}
Let $\mathbb{A},\mathbb{B} \in \mathbb{R}^{m \times m}$ be
positive semi-definite, then
\begin{equation}
\asym{\mathbb{A}+ \mathbb{B}} \leq 
\max\left\{ \asym{\mathbb{A}}, \asym{\mathbb{B}}\right\}.
\end{equation}
\end{proposition}
\begin{proof}
By definition \ref{Apx Asymmetry Index} the matrices
\begin{equation}
\mathbb{J}(s)\equiv
s(\mathbb{A} + \mathbb{A}^t)+i(\mathbb{A} -\mathbb{A}^t)
\quad {{\rm and}} \quad
\mathbb{K}(s)\equiv
s(\mathbb{B} + \mathbb{B}^t)+i(\mathbb{B} -\mathbb{B}^t) 
\end{equation}
with $s \equiv \max\left\{ \asym{\mathbb{A}},\asym{\mathbb{B}}
\right\}$ both are positive semi-definite. It follows that 
\begin{equation}
\mathbb{J}(s)+\mathbb{K}(s) = 
  s\left( \mathbb{A}+\mathbb{B} \right) 
+ s\left( \mathbb{A}+\mathbb{B} \right)^t 
+ i\left(\mathbb{A}+\mathbb{B}\right)
- i\left(\mathbb{A}+\mathbb{B}\right)^t
\end{equation}
is also positive semi-definite and hence $\asym{\mathbb{A}
+\mathbb{B}} \leq s$.
\end{proof}

\begin{proposition}[Dominance of principal submatrices]
\label{Apx Dominance of Principal Submatrices AI}
Let $\mathbb{A}\in\mathbb{R}^{m \times m}$ be positive 
semi-definite and $\bar{\mathbb{A}}\in\mathbb{R}^{p \times p}$ 
$(p < m)$ a principal submatrix of $\mathbb{A}$, then 
\begin{equation}
\asyms{\bar{\mathbb{A}}}\leq\asym{\mathbb{A}}.
\end{equation}
\end{proposition}
\begin{proof}
By definition \ref{Apx Asymmetry Index}
\begin{equation}
\mathbb{K}(s)\equiv \asym{\mathbb{A}}( \mathbb{A} + \mathbb{A}^t)
+i(\mathbb{A} - \mathbb{A}^t)
\end{equation}
is positive semi-definite. Consequently the matrix 
\begin{equation}
\bar{\mathbb{K}}(s)\equiv \asyms{\mathbb{A}} (\bar{\mathbb{A}} 
+ \bar{\mathbb{A}}^t )
+i(\bar{\mathbb{A}} - \bar{\mathbb{A}}^t),
\end{equation}
which constitutes a principal submatrix of $\mathbb{K}$, is also
positive semi-definite and therefore $\asyms{\bar{\mathbb{A}}}\leq
\asym{\mathbb{A}}$.
\end{proof}

\section{Bound on the asymmetry index for special classes of 
matrices}\label{Apx Bound on AI for special classes of matrices}

\begin{theorem}\label{Apx AI Permutation Matrices}
Let $\mathbb{P} \in \left\{0,1 \right\}^{m \times m}$ be a 
permutation matrix and $\id$ the identity matrix, then the matrix
$\mathbbm{1}- \mathbb{P}$ is positive semi-definite on 
$\mathbb{R}^m$ and its asymmetry index fulfils 
\begin{equation}\label{Apx AI Bistoch Mat}
\asym{\mathbbm{1} - \mathbb{P}}
\leq \cot \left( \frac{\pi}{m} \right). 
\end{equation}
\end{theorem}

\begin{proof}

We first show that $\mathbbm{1}-\mathbb{P}$ is positive 
semi-definite. To this end, we note that the matrix elements of 
$\mathbb{P}$ are given by $\left( \mathbb{P} \right)_{ij} =
\delta_{i\pi(j)}$, where $\pi\in S_m$ is the unique permutation
associated with $\mathbb{P}$ and $S_m$ the symmetric group on the
set $\left\{1,\dots,m\right\}$. Now, with $\mathbf{x} \equiv
\left(x_1, \dots,x_m \right)^t \in \mathbb{R}^m$ we have
\begin{eqnarray}
\mathbf{x}^t \left( \mathbbm{1}- \mathbb{P} \right) \mathbf{x}
& = \sum_{i,j=1}^m \left( \delta_{ij} - \delta_{i\pi(j)} \right) x_i x_j
  = \sum_{i,j=1}^m  \frac{\delta_{i\pi(j)}}{2}
    \left(x_i^2 + x_j^2 -2x_i x_j \right)\\
& = \sum_{i,j=1}^m \frac{\delta_{i\pi(j)}}{2}
    \left(x_i -x_j \right)^2 \geq 0.
\end{eqnarray}

We now turn to the second part of Theorem \ref{Apx AI Permutation
Matrices}. For any $\mathbf{z}\equiv (z_1,\dots,z_m)\in
\mathbb{C}^m$ and  $s\geq 0$, we define the quadratic form
\begin{eqnarray}
Q(\mathbf{z},s) & \equiv \mathbf{z}^{\dagger}
\left(s \left( \id-\mathbb{P} \right) +
      s \left( \id-\mathbb{P} \right)^t + 
      i \left( \id - \mathbb{P} \right) -
      i \left( \id - \mathbb{P} \right)^t \right)
\mathbf{z} \label{Apx Q Definition} \\
                & = \mathbf{z}^{\dagger} 
        \left(2s\cdot\id-(s+i)\mathbb{P} - (s-i)\mathbb{P}^t 
        \right)\mathbf{z}.
\end{eqnarray}
By definition \ref{Apx Asymmetry Index} the minimum $s$, for which
$Q(\mathbf{z},s)$ is positive semi-definite, equals the asymmetry
index of $\mathbbm{1}-\mathbb{P}$. This observation enables us to
derive an upper bound for $\asym{\mathbbm{1}-\mathbb{P}}$. To
this end, we make use of the cycle decomposition 
\begin{equation}\label{Apx Cycle Decomposition}
\pi = \left(i_1, \pi(i_1), \dots, \pi^{n_1-1}(i_1) \right) \dots
\left( i_k, \pi(i_k), \dots, \pi^{n_k -1}(i_k) \right)
\end{equation}
of $\pi$, where $i_1, \dots,i_k \in \left\{1, \dots, m \right\}$,
$\pi^l(i)$ is defined recursively by 
\begin{equation}
\pi^l(i) \equiv \pi \left( \pi^{l-1}(i) \right) 
\qquad {{\rm and }} \qquad 
\pi^0(i) =i,
\end{equation}
$k$ denotes the number of independent cycles of and $n_r$ the length
the $r^{{{\rm th}}}$ cycle. By virtue of this decomposition,  
(\ref{Apx Q Definition}) can be rewritten as 
\begin{eqnarray}
Q(\mathbf{z},s) & = \sum_{i,j=1}^m \left( 
2s\delta_{ij} - (s+i) \delta_{i\pi(j)}
              - (s-i) \delta_{\pi(i)j} \right) z_i^{\ast}z_j\\
& = \sum_{i=1}^m 2s z_i^{\ast}z_i - (s+i)z_{\pi(i)}^{\ast}z_{i}
                                  - (s-i)z_{i}^{\ast}z_{\pi(i)}\\
& = \sum_{r=1}^k \sum_{l_r=0}^{n_r-1} 2s z^{\ast}[\pi^{l_r}(i_r)]
z[\pi^{l_r}(i_r)] \nonumber \\
& \hspace*{4cm}  -(s+i)z^{\ast}[\pi^{l_r+1}(i_r)]z[\pi^{l_r}(i_r)]
 \nonumber \\
& \hspace*{4cm}  -(s-i)z^{\ast}[\pi^{l_r}(i_r)]z[\pi^{l_r+1}(i_r)],
\label{Apx Q in PermutForm}
\end{eqnarray}
where, for convenience, we introduced the notation $z[x] \equiv 
z_x$. Next, we define the vectors $\tilde{\mathbf{z}}_r\in
\mathbb{C}^{n_r}$ with elements $\left(\tilde{\mathbf{z}}_r\right)_j
\equiv z\left[\pi^{j-1}(i_r)\right]$ and the Hermitian matrices 
$\mathbb{H}_{n_r}(s)\in\mathbb{C}^{n_r \times n_r}$ with matrix
elements 
\begin{equation}\label{Apx Matrix Elements H}
\left(\mathbb{H}_{n_r}(s)  \right)_{ij}
\equiv 2 s\delta_{ij} - (s+i) \delta_{ij+1} 
- (s-i) \delta_{i+1j},
\end{equation}
where periodic boundary conditions $n_r+1=1$ for the indices $i,j
=1,\dots,n_r$ are understood. These definitions allow us to cast
(\ref{Apx Q in PermutForm}) in the rather compact form
\begin{equation}
Q(\mathbf{z},s) = \sum_{r=1}^k \tilde{\mathbf{z}}_r^{\dagger}
\mathbb{H}_{n_r}(s) \tilde{\mathbf{z}}_r.
\end{equation}
Obviously, any value of $s$ for which all the $\mathbb{H}_{n_r}(s)$
are positive semi-definite serves as a lower bound for $\asym{
\mathbbm{1}-\mathbb{P}}$. Moreover, we can calculate the eigenvalues 
of $\mathbb{H}_{n_r}(s)$ explicitly. Inserting the Ansatz $\mathbf{v}
\equiv (v_1,\dots, v_{n_r})^t\in\mathbb{C}^{n_r}$ into the eigenvalue
equation 
\begin{equation}
\mathbb{H}_{n_r}(s) \mathbf{v} = \lambda \mathbf{v}
\qquad (\lambda \in \mathbb{R})
\end{equation}
yields
\begin{equation}
\lambda v_j = 2sv_j - (s+i) v_{j-1} - (s-i) v_{j+1},
\end{equation}
where again periodic boundary conditions $v_{n_r+1}=v_1$ are
understood. This recurrence equation can be solved by standard
techniques. We put $v_j \equiv \exp \left(2\pi i\kappa j / n_r
\right)$ with $(\kappa = 1, \dots, n_r)$ and obtain the eigenvalues 
\begin{equation}
\lambda_{\kappa} = 2 \left( s- 
s\cos \left( \frac{2 \pi \kappa}{n_r} \right)
-\sin \left( \frac{2 \pi \kappa}{n_r} \right) \right). 
\end{equation}
For any fixed $s\geq 0$, the function 
\begin{equation}
f(x,s) \equiv s - s\cos x - \sin x
\end{equation}
is non-negative for $x \in [x^{\ast},2\pi]$ and strictly negative
for $x\in (0,x^{\ast})$ with
\begin{equation}
x^{\ast} \equiv \arccos \left( \frac{s^2-1}{s^2+1}
\right). 
\end{equation}
Therefore, all the eigenvalues $\lambda_{\kappa}$ of $\mathbb{H}_{
n_r}(s)$ are non-negative, if and only if
\begin{equation}\label{Apx Bound on Eigenvalues}
\frac{2 \pi}{n_r} \geq  \arccos \left(\frac{s^2-1}{s^2+1} \right).
\end{equation}
Solving (\ref{Apx Bound on Eigenvalues}) for $s$ gives the equivalent
condition
\begin{equation}
s \geq \cot \left( \frac{\pi}{n_r} \right).
\end{equation}
Since $n_r \leq m$ and therefore $2 \pi/n_r \geq 2 \pi/m$,
we can conclude that any of the $\mathbb{H}_{n_r}(s)$ is positive
semi-definite for any
\begin{equation}
s \geq \cot \left( \frac{\pi}{m} \right),
\end{equation}
thus establishing the desired result (\ref{Apx AI Bistoch Mat}). \\
\end{proof}

\begin{corollary}\label{Apx AI Doubly Stochastic Matrices}
Let $\mathbb{T} \in\mathbb{R}^{m \times m}$
be doubly stochastic, then the matrix  $\mathbbm{1}- \mathbb{T}$ 
is positive semi-definite and its asymmetry index fulfils
\begin{equation}\label{Apx B Cor1}
\asym{\mathbbm{1} - \mathbb{T}} \leq
\cot \left( \frac{\pi}{m} \right).
\end{equation}
\end{corollary}
\begin{proof}
The Birkhoff-theorem (see p. 549 in \cite{Horn1985}) states
that for any doubly stochastic matrix $\mathbb{T}\in\mathbb{R}^{ 
m\times m}$ there is a finite number of permutation matrices
$\mathbb{P}_1, \dots \mathbb{P}_N \in \{0,1 \}^{m \times m}$ and
positive scalars $\lambda_1, \dots, \lambda_N \in \mathbb{R}$ 
such that 
\begin{equation}
\sum_{k=1}^N \lambda_k=1 
\qquad {{\rm and }} \qquad
\sum_{k=1}^{N} \lambda_k \mathbb{P}_k =\mathbb{T}.
\end{equation}
Hence, we have 
\begin{equation}\label{Apx Decomposition T}
\mathbbm{1}- \mathbb{T} =\sum_{k=1}^N 
\lambda_k \left( \mathbbm{1} - \mathbb{P}_k \right)
\end{equation}
and consequently $\id - \mathbb{T}$ must be positive 
semi-definite by virtue of Theorem
\ref{Apx AI Permutation Matrices}. Furthermore, using
Proposition \ref{Apx Convexity AI} and again Theorem 
\ref{Apx AI Permutation Matrices} gives 
the bound (\ref{Apx B Cor1}).
\end{proof}

\begin{theorem}\label{Apx AI Partial Permutation Matrices}
Let $\bar{\mathbb{P}}\in \left\{0,1\right\}^{m 
\times m}$ be a partial permutation matrix, i.e., any row and 
column of $\bar{\mathbb{P}}$ contains at most one non-zero entry
and all of these non-zero entries are $1$. Then, the matrix
$\id - \bar{\mathbb{P}}$ is positive semi-definite  and its 
asymmetry index fulfils
\begin{equation}\label{Apx AI Partial Permutation Matrices 2}
\asym{ \id - \bar{\mathbb{P}}}
\leq \cot \left( \frac{\pi}{m+1} \right).
\end{equation}
\end{theorem}
\begin{proof}
Let $q$ be the number of non-vanishing entries of
$\bar{\mathbb{P }}$. If $q=0$, $\bar{\mathbb{P}}$ equals the zero
matrix and there is nothing to prove. If $q = m$, $\bar{\mathbb{P
}}$ itself must be a permutation matrix and Lemma 
\ref{Apx AI Permutation Matrices} provides that $\id-
\bar{\mathbb{P}}$ is positive semi-definite as well as the bound 
\begin{equation}
\asym{\id -\bar{\mathbb{P}}}\leq\cot\left( \frac{\pi}{m} \right),
\end{equation}
which is even stronger than
(\ref{Apx AI Partial Permutation Matrices 2}). If 
$0<q<m$, there are two index sets $A\subset \{1,\dots, m\}$
and $B\subset \{1,\dots, m\}$ of equal cardinality $m-q$, such
that the rows of $\bar{\mathbb{P}}$ indexed by $A$ and the
columns of $\bar{\mathbb{P}}$ indexed by $B$ contain only 
zero entries. Clearly, in this case, $\bar{\mathbb{P}}$ is not a 
permutation matrix. Nevertheless, we can define a bijective map 
\begin{equation}
\bar{\pi}: \{1, \dots,m\} \setminus B
\rightarrow 
\{1,\dots,m\}\setminus A
\end{equation}
in such a way that $\bar{\mathbb{P}}$ can be regarded as a 
representation of $\bar{ \pi}$ . To this end, we denote by 
$\{\mathbf{e}_1,\dots, \mathbf{e}_{m}\}$ the canonical basis of
 $\mathbb{R}^{m}$ and define $\bar{\pi}: \; i \mapsto\bar{\pi}(i)$ 
such that
\begin{equation}\label{Apx Definition PiBar}
\bar{\mathbb{P}} \mathbf{e}_i \equiv \mathbf{e}_{\bar{\pi}(i)}.
\end{equation}
This definition naturally leads to the cycle decomposition 
\begin{eqnarray}
\bar{\pi}= &  \left(i_1, \bar{\pi}(i_1), \dots, 
\bar{\pi}^{n_1 -1} (i_1) \right) 
\cdots
\left( i_k, \bar{\pi}(i_{k}), \dots,
\bar{\pi}^{n_{k}-1}(i_{k}) \right) \nonumber \\
& \left[j_1, \bar{\pi}(j_1), \dots, 
\bar{\pi}^{\bar{n}_1 -1} (j_1) \right]
\cdots
\left[j_{\bar{k}}, \bar{\pi}(j_{\bar{k}}), \dots,
\bar{\pi}^{\bar{n}_{\bar{k}}-1}(j_{\bar{k}}) \right].
\label{Apx Cycle Decomposition PiBar}
\end{eqnarray}
Here, we introduced two types of cycles. The ones in round
brackets, which we will term complete, are just ordinary
permutation cycles, which close by virtue of the condition
$\bar{\pi}^{n_r} (i_r) =i_r$ and therefore must be contained 
completely in the set
\begin{equation}
I \equiv \left( \{1,\dots,m\}\setminus B \right) \cap
\left( \{1,\dots,m\}\setminus A \right)
=\{1, \dots, m\}\setminus \left(A \cup B \right).
\end{equation}
The cycles in rectangular brackets, which will be termed  
incomplete, do not close, but begin with a certain $j_{\bar{r}}$ 
taken from the set
\begin{equation}
D \equiv \left( \{1, \dots, m\} \setminus B \right) \setminus
\left(\{1, \dots, m\} \setminus A \right) = 
A \setminus B
\end{equation}
and terminate after $\bar{n}_{\bar{r}} -1$ iterations with
$\bar{\pi}^{\bar{n}_{\bar{r}}-1}(j_{\bar{r}})$, which is 
contained in 
\begin{equation}
R \equiv \left( \{1, \dots, m \}\setminus A \right)\setminus
\left( \{1, \dots, m\} \setminus B \right)
= B\setminus A.
\end{equation}
\begin{figure}[t]
\centering
\epsfig{file=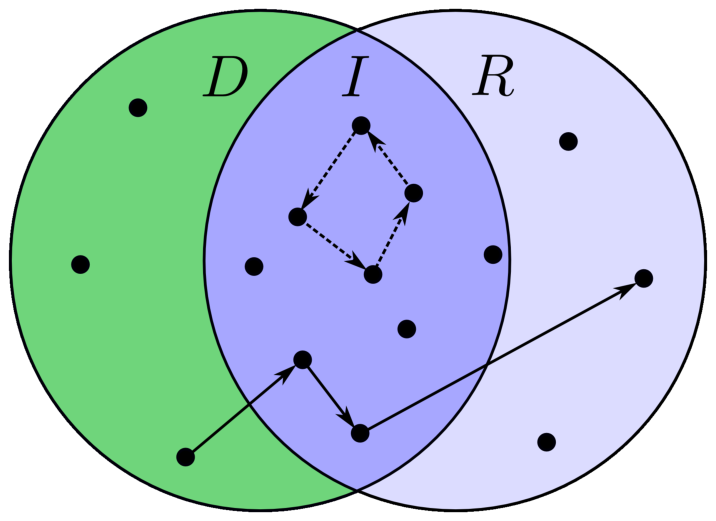, scale=0.7}
\caption{Schematic illustration of the cycle decomposition 
(\ref{Apx Cycle Decomposition PiBar}). The green circle
represents the set $\{1,\dots,m\}\setminus B$, the blue one 
the set $\{1,\dots,m\}\setminus A$. The black dots symbolize 
the elements of the respective sets and the arrows show the
action of the map $\bar{\pi}$. While the dashed arrows form a 
complete cycle, the solid ones combine to an incomplete cycle.
 \label{Fig_Cycle_Sets}}
\end{figure}
Figure (\ref{Fig_Cycle_Sets}) shows a schematic visualization of
 the two different
types of cycles. We note that, since the map $\bar{\pi}$, is 
bijective the cycle decomposition 
\ref{Apx Cycle Decomposition PiBar} is unique up to the choice 
of the $i_r$ and any element of
\begin{equation}
J \equiv \left( \{1, \dots,m\} \setminus B \right) \cup 
\left( \{1, \dots,m\} \setminus A \right)
= \{1, \dots, m\}\setminus\left( A \cap B \right)
\end{equation} shows up exactly once.\\

For the next step, we introduce the vectors 
\begin{equation}
\mathbf{a} \equiv \sum_{i \in A} \mathbf{e}_i 
\qquad {{\rm and}} \qquad
\mathbf{b} \equiv \sum_{i \in B} \mathbf{e}_i,
\end{equation}
as well as the bordered matrix 
\begin{equation}
\mathbb{B} \equiv \left( \! \begin{array}{cc}
\bar{\mathbb{P}} & \mathbf{a}\\
\mathbf{b}^t     & 1-m+q
\end{array}\! \right).
\end{equation}
Obviously, all rows and columns of $\mathbb{B}$ sum up to $1$ and 
all off-diagonal entries are non-negative. Hence, with
$B_{ij}\equiv\left(\mathbb{B}\right)_{ij}$, we have for any
$\mathbf{x}\in\mathbb{R}^m$
\begin{eqnarray}
\mathbf{x}^t\left(\id-\mathbb{B}\right)\mathbf{x} & = 
\sum_{i,j=1}^m \left(\delta_{ij}-B_{ij}\right)x_ix_j =
\sum_{i,j=1}^m \frac{B_{ij}}{2}\left(x_i^2 +x_j^2-2x_ix_j\right)
\nonumber\\
&=\sum_{i,j=1,\;i\neq j}^m \frac{B_{ij}}{2}\left(x_i-x_j\right)^2
\geq 0, \label{Apx PSD B}
\end{eqnarray}
i.e., the matrix $\id-\mathbb{B}$ is positive semi-definite.
Since $\id-\bar{\mathbb{P}}$ is a principal submatrix of
$\id-\mathbb{B}$, (\ref{Apx PSD B}) implies in particular that
$\id-\bar{\mathbb{P}}$ is positive semi-definite, thus establishing 
the first part of Lemma \ref{Apx AI Partial Permutation Matrices}.\\

We will now prove the bound 
(\ref{Apx AI Partial Permutation Matrices 2}) on the asymmetry 
index of $\id-\bar{\mathbb{P}}$. To this end, for any $\mathbf{z}
\in \mathbb{C}^{m+1}$ we associate the matrix $\mathbb{B}$ with 
the quadratic form
\begin{eqnarray}
\bar{Q}(\mathbf{z},s) & \equiv
\mathbf{z}^{\dagger} \left(s(\id-\mathbb{B})+ s(\id-\mathbb{B})^t
+i (\id-\mathbb{B}) - i(\id-\mathbb{B})^t \right) \mathbf{z}\\
& = \mathbf{z}^{\dagger} \left(
2s \cdot \id - (s +i)\mathbb{B} - (s-i) \mathbb{B}^t
\right) \mathbf{z}.\label{Apx Definition QBar}
\end{eqnarray}
and notice that the minimum $s$ for which $\bar{Q}(\mathbf{z
},s)$ is positive semi-definite equals the asymmetry index of 
$\id-\mathbb{B}$. Furthermore, since $\id-\bar{\mathbb{P}}$ is
a principal submatrix of $\id -\mathbb{B}$, Proposition 
\ref{Apx Dominance of Principal Submatrices AI} implies that
this particular value of $s$ is also an upper bound on the 
asymmetry index of $\id-\bar{\mathbb{P}}$. Now, by inserting the
decomposition 
\begin{equation}
\mathbf{z} \equiv \sum_{i=1}^{m+1} z_i \mathbf{e}_i
\end{equation}
into (\ref{Apx Definition QBar}) while keeping in mind the
definition (\ref{Apx Definition PiBar}), we obtain 
\begin{eqnarray}
\bar{Q}(\mathbf{z},s) & = 2 s\sum_{i=1}^{m} z_i^{\ast}z_i
+2 s (m-q) z^{\ast}_{m+1}z_{m+1}
 \nonumber \\
& \quad - (s+i)\left(\sum_{i \in \{1,\dots,m\} \setminus B}
z^{\ast}_{\bar{\pi}(i)} z_i + \sum_{i \in A} z_i^{\ast}
z_{m+1} + \sum_{i \in B} z_{m+1}^{\ast}z_i 
\right) \nonumber \\
& \quad - (s-i)\left(\sum_{i \in \{1,\dots,m\}\setminus B}
z^{\ast}_i z_{\bar{\pi}(i)} + \sum_{i \in A} z_{m+1}^{\ast}
z_i     + \sum_{i \in B} z_{i}^{\ast} z_{m+1}
\right) \label{Apx QBar Elements}.
\end{eqnarray}
By realizing 
\begin{equation}
A = D \cup \left( A \cap B \right),
\quad 
B  = R \cup \left( A \cap B \right),
\quad 
\{1, \dots,m\} = J \cup \left( A \cap B \right)
\end{equation}
and making use of the cycle decomposition 
(\ref{Apx Cycle Decomposition PiBar}), we can rewrite
(\ref{Apx QBar Elements}) as
\begin{eqnarray}
& \bar{Q}(\mathbf{z},s) = 
2s\sum_{r=1}^k \sum_{l_r=0}^{n_r -1} 
z^{\ast}[\bar{\pi}^{l_r}(i_r)]z[\bar{\pi}^{l_r}(i_r)]
+2s\sum_{{\bar{r}}=1}^{\bar{k}} \sum_{l_{\bar{r}}=0}^{n_{\bar{r}}-1}
z^{\ast}[\bar{\pi}^{l_{\bar{r}}}(j_{\bar{r}})]
z[\bar{\pi}^{l_{\bar{r}}}(j_{\bar{r}})] \nonumber\\
& + 2s\sum_{i \in A\cap B} z_i^{\ast}z_i
+ 2s(m-q) z^{\ast}_{m+1}z_{m+1} 
\nonumber \\
& -(s+i) \left( \sum_{r=1}^k \sum_{l_r=0}^{n_r -1} 
z^{\ast}[\bar{\pi}^{l_r+1}(i_r)]z[\bar{\pi}^{l_r}(i_r)]
+\sum_{\bar{r}=1}^{\bar{k}} \sum_{l_{\bar{r}}=0}^{n_{\bar{r}}-2}
z^{\ast}[\bar{\pi}^{l_{\bar{r}}+1}(j_{\bar{r}})]
z[\bar{\pi}^{l_{\bar{r}}}(j_{\bar{r}})] \right. \nonumber\\
& \left. \hspace*{3.0cm}
         + \sum_{i \in D} z_i^{\ast} z_{m+1} 
         + \sum_{i \in R} z^{\ast}_{m+1} z_i 
         + \sum_{i \in A \cap B} \left(
           z_i^{\ast}z_{m+1} + z_{m+1}^{\ast} z_i \right)
\right) \nonumber\\
& -(s-i) \left( \sum_{r=1}^k \sum_{l_r=0}^{n_r -1} 
z^{\ast}[\bar{\pi}^{l_r}(i_r)]z[\bar{\pi}^{l_r+1}(i_r)]
+\sum_{\bar{r}=1}^{\bar{k}} \sum_{l_{\bar{r}}=0}^{n_{\bar{r}}-2}
z^{\ast}[\bar{\pi}^{l_{\bar{r}}}(j_{\bar{r}})]
z[\bar{\pi}^{l_{\bar{r}}+1}(j_{\bar{r}})] \right. \nonumber\\
& \left. \hspace*{3.0cm}
         + \sum_{i \in D} z_{m+1}^{\ast} z_{i} 
         + \sum_{i \in R} z^{\ast}_{i} z_{m+1} 
         + \sum_{i \in A \cap B} \left(
           z_{m+1}^{\ast} z_i +  z_i^{\ast}z_{m+1}  \right)
\right), \nonumber\\
\label{Apx QBar in PBar representation}
\end{eqnarray}
thus explicitly separating contributions from complete and
incomplete cycles. Finally, since we have
\begin{eqnarray}
& \sum_{i \in D} z_i^{\ast}z_{m+1} = \sum_{\bar{r}=1}^{\bar{k}} 
z^{\ast}[j_{\bar{r}}]z_{m+1}, \quad
& \sum_{i \in R} z_{m+1}^{\ast}z_i = \sum_{\bar{r}=1}^{\bar{k}}
z^{\ast}_{m+1}
z[\bar{\pi}^{n_{\bar{r}}-1}(j_{\bar{r}})],\\
& \sum_{i \in D} z_{m+1}^{\ast} z_i = \sum_{\bar{r}=1}^{\bar{k}}
z^{\ast}_{m+1}  z[j_{\bar{r}}], \quad
& \sum_{i \in R} z_i^{\ast}z_{m+1} = \sum_{\bar{r}=1}^{\bar{k}}
z^{\ast}[\bar{\pi}^{n_{\bar{r}}-1}_{\bar{r}}(j_{\bar{r}})]
z_{m+1},
\end{eqnarray}
by employing the definitions 
\begin{eqnarray}
\tilde{\mathbf{z}}_r & \equiv \left( 
z[i_r], z[\bar{\pi}(i_r)], \dots, z[\bar{\pi}^{n_r-1}(i_r)]
\right)^t \in \mathbb{C}^{n_r \times n_r}
\qquad {{\rm and }} \qquad\\
\tilde{\bar{\mathbf{z}}}_{\bar{r}}  & \equiv \left(
z[j_{\bar{r}}], z[\bar{\pi}(j_{\bar{r}})], \dots, 
z[\bar{\pi}^{n_{\bar{r}}-1}(j_{\bar{r}})], z_{m+1} \right)^t
\in \mathbb{C}^{(n_{\bar{r}}+1)\times (n_{\bar{r}}+1)},
\end{eqnarray}
(\ref{Apx QBar in PBar representation}) can be written as
\begin{eqnarray}\label{Apx BarQ in Vector Form}
\bar{Q}(\mathbf{z},s) & = 
\sum_{r=1}^k \tilde{\mathbf{z}}_r^{\dagger} \mathbb{H}_{n_r}(s)
\tilde{\mathbf{z}}_r 
+\sum_{\bar{r}=1}^{\bar{k}} \tilde{\bar{\mathbf{z}}}_{\bar{r}}
^{\dagger} \mathbb{H}_{n_{\bar{r}}+1}(s) \tilde{\bar{\mathbf{z}}}
_{\bar{r}} +2s\sum_{i \in A \cap B} \left\vert
z_i - z_{m+1} \right\vert^2,
\end{eqnarray}
where the matrices $\mathbb{H}_{n}(s)$  are defined in (\ref{Apx
Matrix Elements H}). Since we have already shown for the proof of
Lemma \ref{Apx AI Permutation Matrices}
that $\mathbb{H}_n(s)$ is positive semi-definite for any 
$s \geq \cot \left( \pi / n \right)$, we immediately infer from 
(\ref{Apx BarQ in Vector Form}) that $\bar{Q}(\mathbf{z},s)$ is 
positive semi-definite for any 
\begin{equation}
s\geq\cot \left(\frac{\pi}{{{\rm max}} \{n_r,n_{\bar{r}}+1 \}}
\right).
\end{equation}
Since $\max \{n_r,n_{\bar{r}}+1\}\leq m+1$, we finally end up with
\begin{equation}
\asym{\id - \bar{\mathbb{P}}} \leq
\asym{\id - \mathbb{B}}       \leq
\cot \left( \frac{\pi}{m+1} \right).
\end{equation}
\end{proof}

\begin{corollary}\label{Apx AI Doubly Substochastic Matrices}
Let $\bar{\mathbb{T}}\in\mathbb{R}^{m 
\times m}$ be doubly substochastic, then the matrix
$\id - \bar{\mathbb{T}}$ is positive semi-definite and its 
asymmetry index fulfils
\begin{equation}
\asym{\id - \bar{\mathbb{T}}} \leq
\cot \left( \frac{\pi}{m+1} \right).
\end{equation}
\end{corollary}
\begin{proof}
It can be shown that any doubly substochastic 
matrix is the convex combination of a finite number of partial
permutation matrices $\bar{\mathbb{P}}_k$ (see p. 165 in
\cite{Horn1991}), i.e., we have 
\begin{equation}
\bar{ \mathbb{T}} = \sum_{k=1}^N \lambda_k \bar{\mathbb{P}}_k
\end{equation}
with 
\begin{equation}
\lambda_k >0 
\qquad {{\rm and}} \qquad
\sum_{k=1}^N \lambda_k =1.
\end{equation}
Consequently, it follows 
\begin{equation}
\id - \bar{\mathbb{T}} = \sum_{k=1}^N \lambda_k \left( \id - 
\bar{\mathbb{P}} \right).
\end{equation}
Using the same argument with Lemma 
\ref{Apx AI Partial Permutation Matrices} instead of Lemma
\ref{Apx AI Permutation Matrices} in the proof of Corollary 
\ref{Apx AI Doubly Stochastic Matrices} completes the proof
of Corollary \ref{Apx AI Doubly Substochastic Matrices} .
\end{proof}

\section{Bound on the asymmetry index of the Schur complements}
\label{Apx AI Schur complements}

For $\mathbb{A}\in\mathbb{C}^{m\times m}$ partitioned as 
\begin{equation}\label{Apx Partitioning SC}
\mathbb{A}\equiv \left( \! \begin{array}{cc}
\mathbb{A}_{11} & \mathbb{A}_{12}\\
\mathbb{A}_{21} & \mathbb{A}_{22}
\end{array} \! \right)
\end{equation}
with non-singular $\mathbb{A}_{22}\in\mathbb{R}^{p\times p}$, the
Schur complement of $\mathbb{A}_{22}$ in $\mathbb{A}$ is defined 
by (see p. 18 in \cite{Zhang2005})
\begin{equation}
\mathbb{A}/\mathbb{A}_{22} = \mathbb{A}_{11} - \mathbb{A}_{12}
\mathbb{A}_{22}^{-1}\mathbb{A}_{21}.
\end{equation}
Regarding the asymmetry index, we have the following proposition.

\begin{proposition}[Dominance of the Schur complement]
\label{Apx Dominance of Schur Complement}
Let $\mathbb{A}\in\mathbb{R}^{m\times m}$ be a positive 
semi-definite matrix partitioned as in
(\ref{Apx Partitioning SC}), where $\mathbb{A}_{22}\in
\mathbb{R}^{p\times p}$ is non-singular, then the matrix
$\mathbb{A}/\mathbb{A}_{22}$ is positive semi-definite and
its asymmetry index fulfils
\begin{equation}
\asym{\mathbb{A}/\mathbb{A}_{22}} \leq \asym{\mathbb{A}}.
\end{equation}
\end{proposition}
\begin{proof}
By assumption and by definition \ref{Apx Asymmetry Index}, 
we have for any
$\mathbf{z}\in\mathbb{C}^m$
\begin{equation}
\mathbf{z}^{\dagger} \mathbb{A} \mathbf{z}\geq 0
\qquad {{\rm and}} \qquad
\mathbf{z}^{\dagger}\left(
\asym{\mathbb{A}}(\mathbb{A}+\mathbb{A}^t)
+i(\mathbb{A}-\mathbb{A}^t)\right)\mathbf{z}\geq 0.
\end{equation}
Putting 
\begin{equation}
\mathbf{z} \equiv \left( \! \begin{array}{c} 
\mathbf{z}_p\\ - \mathbb{A}_{22}^{-1}\mathbb{A}_{21}
\mathbf{z}_p \end{array} \! \right) \qquad {{\rm with}} \qquad
\mathbf{z}_p \in \mathbb{C}^p
\end{equation}
yields
\begin{equation}
\mathbf{z}_p^{\dagger} \left( \mathbb{A}/\mathbb{A}_{22} \right)
\mathbf{z}_p \geq 0
\end{equation}
and
\begin{equation}
\mathbf{z}_p^{\dagger} \left( \asym{\mathbb{A}}
(\mathbb{A}/\mathbb{A}_{22}+( \mathbb{A}/\mathbb{A}_{22})^t)
+i(\mathbb{A}/\mathbb{A}_{22}-(\mathbb{A}/\mathbb{A}_{22})^t)
\right)\mathbf{z}_p\geq 0.
\end{equation}
\end{proof}
For the special class of matrices considered in Corollary 
\ref{Apx AI Doubly Substochastic Matrices}, the assertion of 
Proposition \ref{Apx Dominance of Schur Complement} can be even 
strengthened. Before being able to state this stronger result,
we need to prove the following Lemma.

\begin{lemma}\label{Apx AI of the SC of Spcial Matrices}
Let $\bar{\mathbb{T}}\in\mathbb{R}^{m\times m}$ be a doubly 
substochastic matrix and $\mathbb{S}\equiv\id-\bar{\mathbb{T}}$
be partitioned as 
\begin{equation}
\mathbb{S}\equiv\left( \! \begin{array}{cc}
\mathbb{S}_{11} & \mathbb{S}_{12}\\
\mathbb{S}_{21} & \mathbb{S}_{22}
\end{array}\! \right),
\end{equation}
where $\mathbb{S}_{22}\in\mathbb{R}^{p\times p}$ is non-singular,
then there is a doubly substochastic matrix $\bar{\mathbb{T}}_{m-p}
\in\mathbb{R}^{(m-p)\times (m-p)}$, such that 
\begin{equation}
\mathbb{S}/\mathbb{S}_{22} = \id - \bar{\mathbb{T}}_{m-p}.
\end{equation}
\end{lemma}
\begin{proof}
We start with the case $p=1$. Let $\bar{T}_{ij}$ be the matrix 
elements of $\bar{\mathbb{T}}$, then the matrix elements of 
$\mathbb{S}/\mathbb{S}_{22}$ are given by 
\begin{equation}
\left(\mathbb{S}/\mathbb{S}_{22}\right)_{kl}\equiv
\delta_{kl}-\bar{T}_{kl} 
- \frac{\bar{T}_{km}\bar{T}_{ml}}{1-\bar{T}_{mm}}
\end{equation}
with $k,l=1,\dots,m-1$. Obviously, we have
\begin{equation}\label{Apx D Sum Rule 1}
\sum_{k=1}^{m-1} \left(\mathbb{S}/\mathbb{S}_{22}\right)_{kl}
= 1- \sum_{k=1}^{m-1}\bar{T}_{kl} - \frac{\bar{T}_{ml}}{1- 
\bar{T}_{mm}} \sum_{k=1}^{m-1} \bar{T}_{km} \leq 1.
\end{equation}
Furthermore, since by assumption
\begin{equation}
\sum_{i=1}^m \bar{T}_{ij}\leq 1
\end{equation}
it follows 
\begin{equation}\label{Apx D Sum Rule 2}
\sum_{k=1}^{m-1} \left(\mathbb{S}/\mathbb{S}_{22}\right)_{kl}
\geq 1- (1- \bar{T}_{ml}) - \frac{\bar{T}_{ml}(1-\bar{T}_{mm})}
{1-\bar{T}_{mm}}=0.
\end{equation}
Analogously, we find 
\begin{equation}\label{Apx D Sum Rule 3}
0 \leq \sum_{l=1}^{m-1} \left(\mathbb{S}/\mathbb{S}_{22}\right)_{
kl} \leq 1.
\end{equation}
Next, we investigate the sign pattern of the $\left(\mathbb{S}/
\mathbb{S}_{22}\right)_{kl}$. First, for $k\neq l$, we have
\begin{equation}\label{Apx D Constraint 1}
\left(\mathbb{S}/\mathbb{S}_{22}\right)_{kl}
=- \bar{T}_{kl}-\frac{\bar{T}_{km}\bar{T}_{ml}}{1-\bar{T}_{mm}} 
\leq 0.
\end{equation}
Second, we rewrite the $\left(\mathbb{S}/\mathbb{S}_{22}\right)_{
kk}$ as 
\begin{equation}
\left(\mathbb{S}/\mathbb{S}_{22}\right)_{kk} = 
1-\bar{T}_{kk}-\frac{\bar{T}_{km}\bar{T}_{mk}}{1-\bar{T}_{mm}}
=\frac{(1-\bar{T}_{kk})(1-\bar{T}_{mm})
-\bar{T}_{km}\bar{T}_{mk}}{1-\bar{T}_{mm}}
\end{equation}
The numerator appearing on the right hand side can be written as 
\begin{equation}
(1-\bar{T}_{kk})(1-\bar{T}_{mm})
-\bar{T}_{km}\bar{T}_{mk}=
{{\rm Det}}\left( \!\begin{array}{cc}
1-\bar{T}_{kk} & - \bar{T}_{km}\\
-\bar{T}_{mk}   & 1- \bar{T}_{mm}
\end{array}\! \right),
\end{equation}
which is a principal minor of $\id-\bar{\mathbb{T}}$. Since,
by Corollary \ref{Apx AI Doubly Substochastic Matrices} is 
positive semi-definite, we end up with
\begin{equation}\label{Apx D Constraint 2}
0 \leq \left(\mathbb{S}/\mathbb{S}_{22}\right)_{kk} \leq 1.
\end{equation}
From the sum rules (\ref{Apx D Sum Rule 1}), 
(\ref{Apx D Sum Rule 2}) and (\ref{Apx D Sum Rule 3}) and the 
constraints (\ref{Apx D Constraint 1}) and 
(\ref{Apx D Constraint 2}), we deduce that $\id -\mathbb{S}/
\mathbb{S}_{22}$ is doubly substochastic and thus we have proven
Lemma \ref{Apx AI of the SC of Spcial Matrices} for $p=1$.
We now continue by induction. To this end, we assume that Lemma
\ref{Apx AI of the SC of Spcial Matrices} is true for $p=q$. For
$p=q+1$ the matrix $\mathbb{S}_{22}\in\mathbb{R}^{(q+1)\times
(q+1)}$ can be partitioned as 
\begin{equation}
\mathbb{S}_{22}\equiv \left(\! \begin{array}{cc}
W_{11} & \mathbf{W}_{12}^t\\
\mathbf{W}_{21}   & \mathbb{W}_{22}
\end{array}\! \right)
\end{equation}
with $\mathbb{W}_{22}\in\mathbb{R}^{q\times q}$, $W_{11}\in\mathbb{R}$
and accordingly $\mathbf{W}_{12}, \mathbf{W}_{21}\in \mathbb{R}^q$.
The Crabtree-Haynsworth quotient formula (see p. 25 in
\cite{Zhang2005}), allows us to rewrite $\mathbb{S}/\mathbb{S}_{
22}$ as 
\begin{equation}\label{Apx AI for SC quotient formula}
\mathbb{S}/\mathbb{S}_{22} = \left(\mathbb{S}/\mathbb{W}_{22}
\right)/\left(\mathbb{S}_{22}/\mathbb{W}_{22}\right).
\end{equation}
A direct calculation shows that $\mathbb{S}_{22}/\mathbb{W}_{22}\in 
\mathbb{R}$ is the lower right diagonal entry of
$\mathbb{S}/\mathbb{W}_{22}$ (see p. 25 in \cite{Zhang2005} 
for details). Furthermore, by the induction hypothesis, there is 
a doubly substochastic matrix $\bar{\mathbb{T}}_{m-q}\in 
\mathbb{R}^{(m-q)\times (m-q)}$, such that 
\begin{equation}
\mathbb{S}/\mathbb{W}_{22} = \id-\bar{\mathbb{T}}_{m-q}.
\end{equation}
Thus, (\ref{Apx AI for SC quotient formula}) reduces to 
the case $p=1$, for which we have already proven Lemma 
\ref{Apx AI of the SC of Spcial Matrices}.
\end{proof}
From Lemma \ref{Apx AI of the SC of Spcial Matrices} and 
Corollary \ref{Apx AI Doubly Substochastic Matrices}, we 
immediately deduce

\begin{corollary}\label{Apx AI SC of Doubls SS Matrices}
Let $\bar{\mathbb{T}}$, $\mathbb{S}$ and 
$\mathbb{S}_{22}$ be as in Lemma \ref{Apx AI of the SC 
of Spcial Matrices}, then
\begin{equation}
\asym{\mathbb{S}/\mathbb{S}_{22}}\leq \cot \left( 
\frac{\pi}{m-p+1}\right).
\end{equation}
\end{corollary}

\end{document}